\documentclass[aps,prd,10pt,twocolumn,superscriptaddress,nofootinbib,floatfix]{revtex4-2}
\usepackage{lmodern}
\usepackage{CJKutf8}
\usepackage{amsmath,amssymb,amsthm}
\usepackage{graphicx}
\usepackage{orcidlink}
\usepackage{hyperref}
\hypersetup{
  hidelinks,
  pdftitle={How much averaged null energy must an anisotropic bounce violate?},
  pdfauthor={Alan Zhang and Ye Zhou}
}

\newtheorem{theorem}{Theorem}
\newtheorem{lemma}{Lemma}
\newtheorem{proposition}{Proposition}
\newtheorem{corollary}{Corollary}
\theoremstyle{definition}
\newtheorem{definition}{Definition}
\theoremstyle{remark}
\newtheorem{remark}{Remark}

\begin{document}

\title{How much averaged null energy must an anisotropic bounce violate?}

\author{Alan Zhang\,\orcidlink{0009-0000-1187-0755}}
\email{zhana142@mcmaster.ca}
\affiliation{Department of Physics \& Astronomy, McMaster University, Hamilton, Ontario L8S 4M1, Canada}
\affiliation{Perimeter Institute for Theoretical Physics, Waterloo, Ontario N2L 2Y5, Canada}

\author{Ye Zhou (\begin{CJK*}{UTF8}{gbsn}周烨\end{CJK*})\,\orcidlink{0009-0004-7050-7736}}
\email{ye.zhou.horizon@gmail.com}
\affiliation{Independent Researcher, Kunshan, Jiangsu, China}

\date{\today}

\begin{abstract}
We show that in Einstein gravity a null-complete, non-static Bianchi I universe (homogeneous, spatially flat and anisotropic) violates the averaged null energy condition (ANEC) along every complete light ray outside at most four null directions, and along all of them if its total energy density, cosmological constant included, is ever positive. Along each complete ray the averaged null energy is at most $-J_p/(8\pi G)$, with equality whenever the integral converges, where $J_p$ is the squared expansion and shear accumulated by a thin bundle of light rays. With spatial sections $\mathbb R^3$ the achronal ANEC therefore leaves only Minkowski space, as for flat and open Friedmann models. Under the pointwise null energy condition, every null geodesic of a non-static Bianchi I universe is incomplete. The isotropic bounce of effective loop quantum cosmology has $J_p=1.478\,\omega_BH_{\max}$ (light frequency at the bounce times largest Hubble rate), and at the same critical density the cheapest ray of a strongly anisotropic bounce pays four to five times as much. A Casimir-supported turning point on a three-torus has a deficit of the same order, but at Hubble rates well below the Planck scale its light rays circumnavigate the torus and are not achronal, so the achronal ANEC does not apply to them.
\end{abstract}

\maketitle

\section{Introduction}
\label{Sec:Intro}

A cosmology that is to be nonsingular, whether it bounces, emerges from a static state, cycles or is past eternal, must at least be null geodesically complete \cite{BattefeldPeter2015,VilenkinZhang2014,LesnefskyEassonDavies2023,PavlovicSossich2024,EassonLesnefsky2025,KinneyGeneral2026,KhaldiehRosenzweigSteinhardt2026}. The singularity theorems \cite{Penrose1965,HawkingPenrose1970,HawkingEllis2023,Senovilla1998} and the kinematic theorem of Borde, Guth and Vilenkin \cite{BordeVilenkin1994,BordeVilenkin1996,BordeGuthVilenkin2003} say when this fails. For Friedmann--Lema\^{\i}tre--Robertson--Walker (FLRW) models the requirements are well understood. A bounce of a spatially flat model needs a violation of the null energy condition (NEC) \cite{MolinaParisVisser1999}, which is delicate in effective field theory \cite{Rubakov2014,DubovskyGregoireNicolisRattazzi2006}, although the condition might still hold in an averaged sense \cite{Giovannini2017}. Completeness needs a bounce, a loitering phase or an emergent phase \cite{GarciaSaenzHuaZhao2024}. Under the averaged null energy condition (ANEC), non-static flat and open models are incomplete, while the closed branch admits complete solutions \cite{BurwigEassonPRD2026,BurwigEassonAffine2026,EassonLesnefsky2024Inflationary,Easson2026CurvatureBounce}. How much negative null energy a bounce \cite{Moghtaderi2025SNEC} or a Genesis phase \cite{YuZhuCai2026SNEC} may contain has been bounded with the smeared null energy condition (SNEC).

These statements are least developed for anisotropic models, yet shear is the generic outcome of a contracting phase unless an ultra-stiff component suppresses it \cite{BelinskiKhalatnikovLifshitz1970,EricksonWesleySteinhardtTurok2004,CaiBrandenbergerPeter2013}. At the effective bounces of loop quantum cosmology (LQC) with a massless scalar field, shear on average dominates the energy density \cite{McNamaraSainiSingh2023}, although ekpyrotic potentials make the bounces more isotropic \cite{McNamaraSainiSingh2023,BrownMcNamaraSainiSingh2025}. With shear, the null convergence condition (NCC) no longer secures the Borde--Guth--Vilenkin argument \cite{KinneyBGVNCC2026}. For diagonal Bianchi I metrics, Garcia-Saenz, Hua and Sherif showed that the NEC together with expansion in all directions at one time makes the axial null geodesics past incomplete. Histories in which one direction always contracts lie outside their theorem, and the authors close by asking how the ANEC results for FLRW models carry over to anisotropic cosmologies \cite{GarciaSaenzHuaSherif2026}.

This paper answers that question for Bianchi I cosmologies (homogeneous, spatially flat, anisotropic) and makes the answer quantitative. The starting point is an exact identity. In Einstein gravity, the null energy integrated along a complete light ray is at most $-J_p/(8\pi G)$, with equality whenever the integral converges, where $J_p\ge0$ is the squared expansion and shear accumulated along it by a thin bundle of neighboring rays. A non-static universe can keep such a bundle rigid along at most four null directions, so the ANEC fails along every other complete light ray, and along all of them if the total energy density, including the cosmological constant, is ever positive. Averaged conditions along three directions, no two opposite, force a static metric, so with spatial sections $\mathbb R^3$ the achronal ANEC leaves only Minkowski space. This is the anisotropic counterpart of the flat and open FLRW result of Refs.~\cite{BurwigEassonPRD2026,BurwigEassonAffine2026}. The exempt directions require a source that violates the NEC and has nonpositive total energy density at all times. Under the pointwise NEC every null geodesic of a non-static Bianchi I universe is incomplete, even when one direction always contracts. Compact directions cannot remove the ANEC violation; they can only make the violating light rays wind around the torus.

The identity turns the question in the title into a calculation. At each instant the density of $J_p$, per unit squared frequency and averaged over the sky, is $\tfrac29\theta^2+\tfrac7{15}\sigma^2$, with $\theta$ the volume expansion rate and $\sigma$ the shear, so shear adds to the null-energy deficit wherever it is present. For the isotropic bounce of effective LQC with a massless scalar field, $J_p=1.478\,\omega_BH_{\max}$ along every ray, with $\omega_B$ the frequency of the light at the bounce and $H_{\max}$ the largest Hubble rate. At the same critical density, the cheapest ray of a strongly anisotropic bounce pays four to five times the isotropic value, and no ray is exempt. A Casimir-supported turning point on a three-torus \cite{ZeldovichStarobinsky1984} has $J_p=3.496\,\omega_bH_\infty$ in every direction, in terms of the frequency at the turn and the asymptotic Hubble rate, on light rays that all wind around the torus.

The method is that of the Lorentzian splitting theorems \cite{Eschenburg1988,Galloway1989,Newman1990,Galloway2000}, with homogeneity in place of their global hypotheses. The light rays of a Bianchi I spacetime with a given conserved covector form a congruence without caustics whose screen is a fixed plane, an identification implicit in earlier Bianchi I optics \cite{Saunders1968,FleuryPitrouUzan2015}. The Raychaudhuri equation along a complete ray then gives the identity. Its qualitative content, that a nonnegative averaged null curvature forces the expansion and shear of the bundle to vanish, is classical integral focusing \cite{Tipler1978Energy,Tipler1978ODE,ChiconeEhrlich1980,Roman1988,Borde1987,EhrlichKim1994}. For flat FLRW models the identity is that of Ref.~\cite{BurwigEassonAffine2026}, and for the isotropic LQC bounce the integral was evaluated by Li and Zhu \cite{LiZhu2009}. What we add is the fixed screen, which converts vanishing focusing into a frozen metric on that plane; the count of saturated directions, along which a thin bundle stays rigid, and their energy cost; and a pointwise theorem, in which one complete null geodesic with nonnegative null curvature along it and in one transverse null direction at each time forces a static, flat metric. By comparison, Galloway's null splitting theorem assumes null completeness, the NCC and a null line, and yields a totally geodesic null hypersurface, which alone does not force flatness, as plane waves show \cite{Penrose1965PlaneWaves}.

Section~\ref{Sec:NullGeometry} sets up the light bundles and Sec.~\ref{Sec:Averaged} derives the identity. Section~\ref{Sec:Saturation} counts the saturated directions and draws the consequences for the ANEC and for the size of the deficit, which Sec.~\ref{Sec:Examples} evaluates in effective LQC, at a Casimir turning point and on compact sections. Section~\ref{Sec:Splitting} gives the pointwise theorem and Sec.~\ref{Sec:Discussion} discusses sources and limitations; proofs are in the appendices. We call a spacetime null complete if every null geodesic is complete in both affine directions. We set $c=\hbar=1$ and keep $G$; $m_{\mathrm P}^2=1/G$, $\rho_{\mathrm P}=1/G^2$, $\bar M_{\mathrm P}^2=1/(8\pi G)$ and $t_{\mathrm P}=\ell_{\mathrm P}=\sqrt G$.

\section{Light bundles in Bianchi I}
\label{Sec:NullGeometry}

\subsection{Setup}

We consider an $(n+1)$-dimensional Bianchi I spacetime, $n\geq2$, in synchronous coordinates,
\begin{equation}
    ds^2=-dt^2+h_{ij}(t)\,dx^i dx^j,
    \qquad i,j=1,\ldots,n,
    \label{eq:BImetric}
\end{equation}
where $h_{ij}(t)$ is a smooth, symmetric, positive-definite matrix on an open cosmic-time interval $\mathcal T\subseteq\mathbb R$. The physical case is $n=3$. The spatial slices are flat, and the translations generated by $\partial_i$ form an Abelian isometry group \cite{EllisMacCallum1969,EllisVanElst1999,WainwrightEllis1997,PereiraPitrouUzan2007}. Unless stated otherwise the spatial section is $\Sigma=\mathbb R^n$; flat quotients by translations, such as tori, are treated by lifting to this Cartesian cover. We allow $h_{ij}$ to be non-diagonal in every fixed basis \cite{CroppVisser2011}. Define
\begin{equation}
 \begin{aligned}
 K_{ij}&:=\frac12\dot h_{ij},&
 \mathsf K^i{}_j&:=h^{ik}K_{kj},\\
 \mathcal V(t)&:=\sqrt{\det h(t)},&
 \frac{\dot{\mathcal V}}{\mathcal V}
   &=\operatorname{tr}\mathsf K ,
 \end{aligned}
    \label{eq:Kdef}
\end{equation}
where a dot denotes $d/dt$. The operator $\mathsf K$ is self-adjoint with respect to $h$, and its eigenvalues $\kappa_1\le\cdots\le\kappa_n$ are the principal expansion rates. We write $\theta=\operatorname{tr}\mathsf K$ for the volume expansion rate, $\sigma=\mathsf K-(\theta/n)\mathsf 1$ for the shear and $\sigma^2:=\operatorname{tr}(\sigma^2)$. With the curvature convention
\begin{equation}
    [\nabla_\mu,\nabla_\nu]V^\rho
    =R^\rho{}_{\sigma\mu\nu}V^\sigma,
    \qquad
    R_{\mu\nu}=R^\rho{}_{\mu\rho\nu},
    \label{eq:curvatureconvention}
\end{equation}
the Ricci tensor in the fixed coordinate basis is
\begin{equation}
 \begin{aligned}
 R_{00}&=-\frac{d}{dt}\operatorname{tr}\mathsf K
          -\operatorname{tr}(\mathsf K^2),\qquad R_{0i}=0,\\
 R_{ij}&=\dot K_{ij}+(\operatorname{tr}\mathsf K)K_{ij}
       -2K_{im}\mathsf K^m{}_j,
 \end{aligned}
 \label{eq:RicciGeneral}
\end{equation}
and the Hamiltonian constraint of the flat slices is the identity
\begin{equation}
 2G_{00}=(\operatorname{tr}\mathsf K)^2-\operatorname{tr}(\mathsf K^2).
 \label{eq:Hamiltonian}
\end{equation}
In Einstein gravity $G_{00}=8\pi G\rho+\Lambda$, and non-diagonal evolution needs anisotropic stress. For isotropic stress $d(\mathcal V\sigma)/dt=0$, so $h$ is diagonal in a constant basis.

\subsection{Fixed-covector congruences}
\label{Sec:OpticalGeometry}

Let $k^\mu=dx^\mu/d\lambda$ be the tangent of a future-directed, affinely parametrized null geodesic. Translation invariance implies conservation of
\begin{equation}
    p_i:=g_{i\mu}k^\mu=h_{ij}\frac{dx^j}{d\lambda},
    \qquad
    \frac{dp_i}{d\lambda}=0 .
    \label{eq:pconserved}
\end{equation}
The covector $p$ is nonzero for every null geodesic. With
\begin{equation}
    \omega_p(t):=\left(h^{ij}(t)p_i p_j\right)^{1/2},
    \label{eq:omegadef}
\end{equation}
the frequency measured by comoving observers, the geodesic equations reduce to
\begin{equation}
    \frac{dt}{d\lambda}=\omega_p,
    \qquad
    \frac{dx^i}{d\lambda}=h^{ij}p_j .
    \label{eq:nullgeodesics}
\end{equation}
A positive rescaling of $p$ changes only the affine normalization. We write $[p]_+=\{cp:c>0\}$ for the oriented ray and $[p]=\{cp:c\neq0\}$ for the unoriented line of a covector; $p$ and $-p$ label opposite rays on the same line. The comoving unit direction of propagation,
\begin{equation}
    q^i:=\frac{h^{ij}p_j}{\omega_p},
    \qquad
    \frac{\dot\omega_p}{\omega_p}=-K_{ij}q^iq^j,
    \label{eq:qdef}
\end{equation}
generally rotates in time although $p$ is fixed.

For fixed $p$ the geodesics \eqref{eq:nullgeodesics} form a homogeneous congruence. On the Cartesian cover the function
\begin{equation}
    u_p(t,x)=p_i x^i-\int^t\omega_p(s)\,ds
    \label{eq:opticalfunction}
\end{equation}
satisfies the eikonal equation with $k_\mu=\nabla_\mu u_p=(-\omega_p,p_i)$, so the congruence is affinely parametrized and twist-free. Its screen, the space of transverse separations between neighboring rays, is the $h$-orthogonal complement of $q$. Since $h(q,v)=p(v)/\omega_p$ for every spatial vector $v$,
\begin{equation}
    q^\perp=\ker p:=\{v^i:\ p_i v^i=0\}.
    \label{eq:screenkerp}
\end{equation}
The screen is therefore a fixed hyperplane of the translation space, even though $q$ moves. Because $k_i=p_i$ is spatially constant, $\nabla_jk_i=\omega_pK_{ij}$, and the optical tensor on the screen, which describes the expansion and shear of a thin bundle of these rays, is
\begin{equation}
    B_p
    =\omega_p\,
    \mathsf{P}\mathsf{K}\mathsf{P}\big|_{\ker p},
    \qquad
    \mathsf P=\mathsf 1-q\otimes q^\flat ,
    \label{eq:Bmaster}
\end{equation}
a self-adjoint endomorphism of $(\ker p,h)$ with trace
\begin{equation}
    \Theta_p=\omega_p\left(\operatorname{tr}\mathsf K-K_{ij}q^iq^j\right)
    =\frac{d}{d\lambda}\ln\!\left(\mathcal V\omega_p\right).
    \label{eq:Thetadef}
\end{equation}
Null propagation and the Sachs optics of diagonal Bianchi I models go back to Saunders and have been developed through the Jacobi map and the optical shear \cite{Saunders1968,FleuryPitrouUzan2015,FleuryNugierFanizza2016}; Eq.~\eqref{eq:Bmaster} holds for a general spatial metric and keeps the affine normalization explicit. By \eqref{eq:screenkerp} and \eqref{eq:Bmaster}, at each time
\begin{equation}
 B_p(t)=0\quad\Longleftrightarrow\quad\dot h(t)\big|_{\ker p\times\ker p}=0 .
 \label{eq:hyperplanefreezing}
\end{equation}
The congruence has no caustics. Every inextendible generator traverses all of $\mathcal T$, with affine parameter
\begin{equation}
 \lambda(t)=\lambda_0+\int_{t_0}^{t}\frac{ds}{\omega_p(s)},
 \label{eq:affinetraversal}
\end{equation}
and it is complete in both directions if and only if both endpoint integrals of $1/\omega_p$ diverge (Lemma~\ref{lem:fixedp} in Appendix~\ref{App:Rigidity}). Completeness of one generator therefore implies it for all generators with the same covector.

\subsection{The optical bulk}
\label{Sec:AffineIdentity}

For the fixed-$p$ congruence the Raychaudhuri equation reads \cite{Raychaudhuri1955,Sachs1961}
\begin{equation}
 \begin{aligned}
 \frac{d\Theta_p}{d\lambda}
   &=-\operatorname{tr}(B_p^2)-R_{\mu\nu}k^\mu k^\nu,\\
 \operatorname{tr}(B_p^2)
   &=\operatorname{tr}(\sigma_p^2)+\frac{\Theta_p^2}{m}
     \ge\frac{\Theta_p^2}{m},
 \end{aligned}
    \label{eq:Raychaudhuri}
\end{equation}
where $m:=n-1$ is the screen dimension and $\sigma_p$ is the trace-free part of $B_p$. Integrating over a finite affine interval gives the sum rule
\begin{equation}
 \begin{aligned}
  \int_{\lambda_1}^{\lambda_2}R_{\mu\nu}k^\mu k^\nu\,d\lambda
  &=\Theta_p(\lambda_1)-\Theta_p(\lambda_2)\\
  &\quad-\int_{\lambda_1}^{\lambda_2}\operatorname{tr}(B_p^2)\,d\lambda ,
 \end{aligned}
    \label{eq:affinesumrule}
\end{equation}
which splits the integrated null curvature into a boundary term of indefinite sign and a nonpositive bulk term. We call the integral in the bulk term the optical bulk. In cosmic time it is
\begin{equation}
    \int\operatorname{tr}(B_p^2)\,d\lambda
    =\int\omega_p\,\operatorname{tr}\!\left[\left(\mathsf{P}\mathsf{K}\mathsf{P}\big|_{\ker p}\right)^2\right]dt ,
    \label{eq:Jcosmictime}
\end{equation}
and we write $J_p$ for its value over the whole generator. Under $p\mapsto cp$ with $c>0$ it scales as $J_p\mapsto cJ_p$, so a frequency normalization must be stated with every value. For a diagonal metric $h=\operatorname{diag}(a_1^2,\ldots,a_n^2)$ with $H_i=\dot a_i/a_i$,
\begin{equation}
 \omega_p^2=\sum_i\frac{p_i^2}{a_i^2},
 \qquad
 \nu_i=\frac{p_i}{a_i\omega_p}
 \label{eq:diagonaldirection}
\end{equation}
are the frequency and the components of $q$ in the orthonormal frame, and
\begin{equation}
 \frac{\operatorname{tr}(B_p^2)}{\omega_p^2}
 =\sum_iH_i^2-2\sum_iH_i^2\nu_i^2
   +\Bigl(\sum_iH_i\nu_i^2\Bigr)^{2} .
 \label{eq:diagonalbulk}
\end{equation}
This expression is all that is needed to evaluate $J_p$ for a given diagonal background. For a flat FLRW metric it reduces to $(n-1)H^2$, so that $J_p=2|p|\int H^2a^{-1}dt$ in $3+1$ dimensions, the bulk term of the affine identity of Ref.~\cite{BurwigEassonAffine2026}.

\subsection{Saturated directions}
\label{Sec:SaturationDef}

We call a ray $[p]_+$ \emph{saturated} if $B_p(t)=0$ for every $t\in\mathcal T$, so that thin bundles of its light rays never expand or shear; by \eqref{eq:hyperplanefreezing} this holds if and only if $\dot h(t)$ vanishes on $\ker p\times\ker p$ at all times. A line $[p]$ is saturated if one, and hence both, of its rays are saturated, and an instant $t_*$ with $B_p(t_*)=0$ is a \emph{saturated instant} of $[p]$ (Definition~\ref{def:saturation} in Appendix~\ref{App:Rigidity}). Along a saturated ray $\Theta_p=0$, so $\mathcal V\omega_p$ is constant, and \eqref{eq:Raychaudhuri} gives $R_{\mu\nu}k^\mu k^\nu=0$; the screen components of the tidal tensor vanish as well. A ray is saturated if and only if $J_p=0$.

A saturated direction freezes the induced metric on its screen, including the shear degrees of freedom. If $\dot h$ vanishes on $\ker p$ for all $t$, then $\dot h=p\odot\eta(t)$ for a covector $\eta(t)$, where $p\odot\eta:=p\otimes\eta+\eta\otimes p$ (Lemma~\ref{lem:singlehyperplane} in Appendix~\ref{App:Rigidity}). In constant coordinates with $p=dz$, after a constant linear change of the coordinates on $\ker p$, the metric takes the normal form
\begin{equation}
 h=\sum_{A=1}^{n-1}\bigl(dx^A+b_A(t)\,dz\bigr)^2+a^2(t)\,dz^2,
 \qquad a>0 .
 \label{eq:normalform}
\end{equation}
Only the scale factor $a(t)$ of the propagation direction and the tilt $b(t)$, which shears that direction against the frozen screen, can depend on time. For this metric
\begin{equation}
 \omega_p=\frac1a,
 \qquad
 q=\frac1a\left(\partial_z-b^A\partial_A\right),
 \qquad
 d\lambda=a\,dt,
 \label{eq:normalaffine}
\end{equation}
and $\{E_A=\partial_A,\,q\}$ is an orthonormal spatial frame in which
\begin{equation}
 K=
 \begin{pmatrix}
  0 & v\\
  v^{\mathsf T} & H
 \end{pmatrix},
 \qquad
 v_A:=\frac{\dot b_A}{2a},
 \qquad
 H:=\frac{\dot a}{a}.
 \label{eq:normalK}
\end{equation}

\section{The averaged-null-energy identity}
\label{Sec:Averaged}

Pointwise conditions are too strong for quantum matter, which can violate the NEC locally \cite{KontouSanders2020}; the appropriate hypotheses are averages along a light ray. In this section $(M^d,g)$ is a general spacetime, $d=n+1\ge3$, and $\gamma:\mathbb R\to M$ is a null geodesic that is complete in both directions. We assume that $\gamma$ is an integral curve, with tangent $k=\dot\gamma$, of a smooth, affinely parametrized, twist-free null congruence defined on a neighborhood of $\gamma(\mathbb R)$; we call such a congruence regular along $\gamma$. Its optical tensor $B$ is then smooth along $\gamma$, and only the Raychaudhuri equation \eqref{eq:Raychaudhuri} with $m=d-2$ is used. In Bianchi I a regular congruence exists for every generator (Sec.~\ref{Sec:OpticalGeometry}). Define
\begin{equation}
 \begin{aligned}
 I(U,V)&:=\int_{-U}^{V}R_{\mu\nu}k^\mu k^\nu\,d\lambda,\\
 J(U,V)&:=\int_{-U}^{V}\operatorname{tr}(B^2)\,d\lambda,
 \end{aligned}
 \label{eq:IJdef}
\end{equation}
for $U,V>0$, the total optical bulk $J^{\mathrm{tot}}:=\int_{\mathbb R}\operatorname{tr}(B^2)\,d\lambda\in[0,\infty]$,
and the independent lower limit
\begin{equation}
 \underline I:=\lim_{R\to\infty}\ \inf_{U,V\ge R}I(U,V)\in[-\infty,\infty].
 \label{eq:cutoffliminf}
\end{equation}
The condition $\underline I\ge0$ says that for every $\epsilon>0$ the cutoff integrals exceed $-\epsilon$ once both cutoffs are large. It is invariant under shifts of the affine origin and under positive affine rescalings, and it is the two-ended form used in classical integral-focusing theorems \cite{Tipler1978Energy,Borde1987} and in Verch's two-dimensional ANEC \cite{Verch2000}. When $\lim I(U,V)$ exists as $U,V\to\infty$ independently we denote it by $I^{\mathrm{ind}}$. For the fixed-$p$ congruence of a Bianchi I spacetime we write $I_p(U,V)$ for the cutoff integrals, $I_p$ and $J_p$ for $I^{\mathrm{ind}}$ and $J^{\mathrm{tot}}$, and keep the superscripts where the prescription matters.

For general integrands the limiting prescription matters; the Raychaudhuri structure removes most of this freedom. We call a cutoff path $s\mapsto(U(s),V(s))$ admissible if both cutoffs are nondecreasing, tend to infinity, and advance together over an unbounded range of the parameter; the precise condition is Eq.~\eqref{eq:admissible} of Appendix~\ref{App:AveragedProofs}. Straight lines and power laws in the $(U,V)$ plane are admissible. Paths that move one cutoff at a time are not, and their exclusion is necessary (Appendix~\ref{App:Staircase}).

\begin{proposition}[Path sum rule]
\label{thm:path}
Under the hypotheses above, every admissible cutoff path satisfies
\begin{equation}
 \liminf_{s\to\infty}I\bigl(U(s),V(s)\bigr)\le-J^{\mathrm{tot}},
 \label{eq:pathliminf}
\end{equation}
where the right-hand side is $-\infty$ if $J^{\mathrm{tot}}=\infty$. Consequently $\underline I\le-J^{\mathrm{tot}}$, any limit of $I$ along an admissible path equals $-J^{\mathrm{tot}}$, and $I^{\mathrm{ind}}=-J^{\mathrm{tot}}$ whenever the independent limit exists. If $\underline I\ge0$, or if the lower limit along a single admissible path is nonnegative, then $B\equiv0$ and $R_{\mu\nu}k^\mu k^\nu\equiv0$ along $\gamma$.
\end{proposition}

In Einstein gravity $R_{\mu\nu}k^\mu k^\nu=8\pi G\,T_{\mu\nu}k^\mu k^\nu$ for any cosmological constant, so along a complete generator of a Bianchi I spacetime on which the null-energy integral converges,
\begin{equation}
 \int_{\mathbb R}T_{\mu\nu}k^\mu k^\nu\,d\lambda
 =\frac{I_p}{8\pi G}=-\frac{J_p}{8\pi G},
 \qquad
 J_p=\int_{\mathbb R}\operatorname{tr}(B_p^2)\,d\lambda .
 \label{eq:budgetidentity}
\end{equation}
When the integral does not converge, its lower limit along every admissible cutoff path is still at most $-J_p/(8\pi G)$. The averaged null energy along a complete light ray is minus the squared expansion and shear accumulated by a thin bundle traveling with it, and $J_p/(8\pi G)$ is the deficit that any Einsteinian source of the geometry must supply along that ray. Since $T_{0i}=0$ and $d\lambda=dt/\omega_p$, the same statement in terms of the energy density $\rho$ and the stress $P_{ij}:=T_{ij}$ of the source reads
\begin{equation}
 8\pi G\int_{\mathcal T}\omega_p\bigl(\rho+P_{ij}q^iq^j\bigr)\,dt=-J_p ,
 \label{eq:modelbuilder}
\end{equation}
which is the form a matter model for a given background has to satisfy along every complete ray on which the integral converges; a cosmological constant drops out of the left-hand side.

The last, qualitative statement of Proposition~\ref{thm:path} is the classical integral-focusing result \cite{Tipler1978Energy,Tipler1978ODE,ChiconeEhrlich1980,Borde1987,EhrlichKim1994,BeemEhrlichEasley1996}: a complete null geodesic along which a regular congruence exists cannot carry a nonnegative two-sided null-curvature integral unless its optical tensor vanishes. Proposition~\ref{thm:path} supplies the upper bound $-J^{\mathrm{tot}}$ for every admissible path, including $J^{\mathrm{tot}}=\infty$, and identifies any existing limit as $-J^{\mathrm{tot}}$, without curvature bounds. The identity $I^{\mathrm{ind}}=-J^{\mathrm{tot}}$ is a one-geodesic analogue of the Hopf--Green identity for closed Riemannian manifolds without conjugate points \cite{Hopf1948,Green1958}, and it extends the affine identity used for FLRW models in Ref.~\cite{BurwigEassonAffine2026} to congruences with shear.

Smooth weights, the form used in proofs of the ANEC for quantum fields \cite{FewsterOlumPfenning2007}, give an exact statement at every scale. Let $\chi\in C_c^\infty(\mathbb R)$ be real with $\chi(0)=1$, and for a smearing length $\tau>0$ set $\chi_\tau(\lambda)=\chi(\lambda/\tau)$, $C_\chi=\int_{\mathbb R}\chi'(s)^2\,ds$ and $I_\chi(\tau):=\int_{\mathbb R}\chi_\tau^2\,R_{\mu\nu}k^\mu k^\nu\,d\lambda$. Then
\begin{align}
 I_\chi(\tau)&=\frac{mC_\chi}{\tau}-\int_{\mathbb R}\operatorname{tr}\Bigl[\bigl(\chi_\tau B-\chi_\tau'\mathsf 1\bigr)^2\Bigr]d\lambda ,
 \label{eq:smearedidentity}\\
 \lim_{\tau\to\infty}I_\chi(\tau)&=-J^{\mathrm{tot}}\in[-\infty,0]
 \label{eq:smearedlimit}
\end{align}
for every $\tau>0$ (Proposition~\ref{thm:smeared} in Appendix~\ref{App:AveragedProofs}). Discarding the last term in \eqref{eq:smearedidentity} gives $I_\chi(\tau)\le mC_\chi/\tau$, the focusing estimate behind singularity theorems with smeared energy hypotheses \cite{FewsterGalloway2011,FewsterKontou2020,FewsterKontou2022,FreivogelKontouKrommydas2022}; the identity keeps that term.

In what follows, an \emph{averaged condition} along a complete generator of a Bianchi I spacetime is any one of the following: $\underline I_p\ge0$; a nonnegative lower limit of $I_p$ along one admissible cutoff path; or $\lim_{\tau\to\infty}I_{p,\chi}(\tau)\ge0$ for one weight $\chi$. By Proposition~\ref{thm:path} and Eq.~\eqref{eq:smearedlimit}, each of them forces $B_p\equiv0$ and $R_{\mu\nu}k^\mu k^\nu\equiv0$ on $\mathcal T$, so that the line $[p]$ is saturated (Corollary~\ref{cor:bianchiaveraged} in Appendix~\ref{App:AveragedProofs}). In Einstein gravity these are the corresponding forms of the ANEC along the ray.

\paragraph*{Remark on the SNEC.}
The proposed smeared null energy condition, $\int f^2\langle T_{\mu\nu}\rangle k^\mu k^\nu\,d\lambda\ge-(4\mathcal B/G)\int f'^2\,d\lambda$ with a constant $\mathcal B$ \cite{FreivogelKrommydas2018,FreivogelKontouKrommydas2022}, reads $I_\chi(\tau)\ge-32\pi\mathcal B\,C_\chi/\tau$ for $f=\chi_\tau$ after the Einstein null projection. By \eqref{eq:smearedidentity} with $m=2$ it is equivalent to
\begin{equation}
 \int_{\mathbb R}\operatorname{tr}\Bigl[\bigl(\chi_\tau B-\chi_\tau'\mathsf 1\bigr)^2\Bigr]d\lambda
 \le\bigl(2+32\pi\mathcal B\bigr)\frac{C_\chi}{\tau} .
 \label{eq:SNECequivalence}
\end{equation}
Nothing has been dropped, so along a given ray the SNEC at smearing length $\tau$ is a condition on the profile of the optical tensor alone. Both sides scale as $1/s$ under $\lambda\mapsto s\lambda$, so the condition tests the dimensionless shape of a bounce against $\tau$ in units of its affine duration, independently of its curvature in Planck units. The left-hand side tends to $J^{\mathrm{tot}}$ as $\tau\to\infty$ while the right-hand side tends to zero, so a ray with $J^{\mathrm{tot}}>0$ violates the condition once $\tau$ is large enough.

The isotropic LQC bounce of Sec.~\ref{Sec:LQC} is the member $q=1/3$ of the family $a\propto(1+t^2/t_b^2)^{q/2}$ tested in Ref.~\cite{Moghtaderi2025SNEC}. For a Gaussian window centered on the bounce and $\mathcal B=1/(32\pi)$, the value used in Refs.~\cite{FreivogelKontouKrommydas2022,Moghtaderi2025SNEC,YuZhuCai2026SNEC}, the condition fails for $\tau>\tau_c=1.53\,\lambda_b$, where $\lambda_b:=t_b/\omega_B$ and $t_b$ is the cosmic time at which the Hubble rate peaks. The critical length is $1.25$ times the affine curvature radius $(\max|R_{\mu\nu}k^\mu k^\nu|)^{-1/2}$, and other windows move it between $0.76$ and $6.1\,\lambda_b$ (Appendix~\ref{App:LQC}). The condition was proposed for smearing lengths small compared with the curvature radius and for curvature far below the Planck scale \cite{FreivogelKrommydas2018}, possibly extending to lengths of order the curvature radius \cite{FreivogelKontouKrommydas2022}. For LQC itself the crossing is Planckian, $t_b\approx0.18\,t_{\mathrm P}$, and the test does not apply. It speaks only to bounces of this shape at lower density, for which the crossing lies between $0.6$ and $5$ curvature radii, depending on the window.

\section{Saturated directions and the ANEC}
\label{Sec:Saturation}

An averaged condition along one complete direction freezes one screen and leaves the normal form \eqref{eq:normalform}. We now show that every instant at which some $B_p$ vanishes has $G_{00}\le0$, that is $8\pi G\rho+\Lambda\le0$ in Einstein gravity, and that a non-static metric can freeze at most two screens. The averaged null curvature is therefore strictly negative in every other complete direction, and averaged conditions along three directions, no two of them opposite, force staticity.

\subsection{The energy cost of saturation}

\begin{lemma}[Energy density at a saturated instant]
\label{lem:energy}
If $B_p(t_*)=0$ for some covector $p$ at some time $t_*$, that is, if $t_*$ is a saturated instant of $[p]$, then
\begin{equation}
 G_{00}(t_*)=-|v(t_*)|^2\le0,
 \qquad
 v_A=K(E_A,q),
 \label{eq:energylemma}
\end{equation}
where $\{E_A\}$ is an orthonormal basis of $\ker p$. In Einstein gravity, $8\pi G\rho+\Lambda=-|v|^2$ at $t_*$.
\end{lemma}

Thus, if the total energy density including the cosmological constant is positive at some time, every fixed-covector congruence has $B_p\ne0$ at that time, so no direction is saturated. A light bundle can stay rigid at an instant only if the total energy density is nonpositive then. Since only $8\pi G\rho+\Lambda$ enters, a negative cosmological constant can supply this sign. The substantive cost of a direction that stays saturated throughout a null-complete, non-static history is a violation of the NEC. The generators of such a direction are complete and have $R_{\mu\nu}k^\mu k^\nu\equiv0$, so by Theorem~\ref{thm:main} of Sec.~\ref{Sec:Splitting} there is an instant at which the null curvature is negative in every null direction $\partial_t+e$ with $e$ a unit vector in the screen. At an isotropic turning point, where $\mathsf K=0$, every direction has a saturated instant and the lemma gives $G_{00}=0$.

\subsection{At most four saturated rays}
\label{Sec:GlobalSaturation}

\begin{proposition}[Saturated-line classification]
\label{prop:classification}
Let $h(t)$ be a smooth non-static Bianchi I metric in any dimension $n\ge2$. Its set of saturated lines contains zero, one or two elements, that is zero, two or four oriented rays. If it contains two lines $[p]$ and $[r]$, the metric has the form
\begin{equation}
  h(t)=h_0+\Phi(t)\,p\odot r,
  \qquad h_0=h(t_0),
  \label{eq:twolinefamily}
\end{equation}
with a scalar function $\Phi$, and the saturated rays are exactly $[\pm p]_+$ and $[\pm r]_+$. A static metric saturates every line.
\end{proposition}

Two frozen screens leave a single function of time, and a third frozen screen leaves none. The two-line metrics are diagonal in a constant basis, with one scale factor growing while another shrinks so that both screens stay rigid; a single saturated line, by contrast, allows evolution that is non-diagonal in every fixed basis (Sec.~\ref{Sec:SharpExamples}). At a fixed time the directions with $B_p=0$ are classified by $\mathsf K$ alone: a rank-one $\mathsf K$ gives one antipodal pair, an indefinite rank-two $\mathsf K$ gives four directions, and every other nonzero $\mathsf K$ gives none (Appendix~\ref{App:SaturationProofs}). Along a saturated ray the null hypersurfaces $u_p=\mathrm{const}$ are totally geodesic and $k$ is a Weyl-aligned null direction, so the count is consistent with the at most four principal null directions of a nonzero Weyl tensor in four dimensions \cite{Stephani2003,OrtaggioPravdaPravdova2013}; Proposition~\ref{prop:classification} is global in time and holds at conformally flat events and in every dimension.

\subsection{The ANEC off the saturated set}

\begin{corollary}[Averaged null curvature off the saturated set]
\label{cor:strict}
In a non-static Bianchi I spacetime, let $\gamma$ be a two-sided complete null geodesic whose covector does not lie on one of the at most two saturated lines. Then $0<J_p^{\mathrm{tot}}\le\infty$, and
\begin{equation}
 \underline I_p\le-J_p^{\mathrm{tot}}<0,
 \qquad
 \lim_{\tau\to\infty}I_{p,\chi}(\tau)=-J_p^{\mathrm{tot}},
 \label{eq:strict}
\end{equation}
with the same upper bound for the lower limit along every admissible cutoff path. In Einstein gravity the averaged null energy along $\gamma$ is strictly negative in each of these prescriptions. Consequently a null-complete, non-static Bianchi I spacetime violates the averaged null convergence condition, $\underline I_p\ge0$ or one of its variants in Sec.~\ref{Sec:Averaged}, along every null geodesic outside at most four oriented directions, and in Einstein gravity it violates the ANEC there.
\end{corollary}

Along saturated rays, in contrast, $\Theta_p=R_{\mu\nu}k^\mu k^\nu=0$ pointwise, so every cutoff integral vanishes. Saturated lines are rare. By Lemma~\ref{lem:energy} and the pointwise classification, no line is saturated if at one instant $G_{00}>0$ (in Einstein gravity, $8\pi G\rho+\Lambda>0$), or $\mathsf K$ has rank at least three, or $\mathsf K$ is a semidefinite operator of rank two. For a diagonal metric in $3+1$ dimensions a saturated direction requires one scale factor to be constant in time (Appendix~\ref{App:SaturationProofs}), so a history in which all three scale factors vary has none, whatever the sign of its energy density. In all these cases a null-complete, non-static Bianchi I spacetime violates the ANEC on every complete null geodesic.

\subsection{Averaged rigidity from three directions}
\label{Sec:AveragedRigidity}

Achronal averaged conditions apply only to null lines, inextendible null geodesics on which no two points can be joined by a timelike curve. On sections $\mathbb R^n$ the optical function \eqref{eq:opticalfunction} decreases along every future-directed timelike curve, so every two-sided complete generator is a null line (Lemma~\ref{lem:nulllines} in Appendix~\ref{App:Toroidal}).

\begin{corollary}[Averaged rigidity]
\label{cor:averagedrigidity}
Suppose three pairwise distinct covector lines have two-sided complete generators along which one of the averaged conditions of Sec.~\ref{Sec:Averaged} holds. Then $h$ is constant and $\mathcal T=\mathbb R$. In Einstein gravity the same conclusion follows from the corresponding averaged NEC for $T_{\mu\nu}k^\mu k^\nu$. In particular, if $\Sigma=\mathbb R^n$ and the spacetime is null complete, then every null geodesic is a null line by Lemma~\ref{lem:nulllines}, and an achronal averaged NEC on all null lines, in one of these prescriptions, makes the spacetime Minkowski space after a constant linear change of spatial coordinates.
\end{corollary}

This is the Bianchi I counterpart of the ANEC analysis of flat FLRW models by Burwig and Easson \cite{BurwigEassonPRD2026,BurwigEassonAffine2026}, and it answers, for Bianchi I, the question raised in Ref.~\cite{GarciaSaenzHuaSherif2026}: under an achronal averaged NEC no non-static, null-complete Bianchi I spacetime with sections $\mathbb R^n$ survives. Three directions are the minimum, because two lines can be saturated. Unlike a flat FLRW model, a null-complete, non-static anisotropic metric can therefore obey the averaged null convergence condition on exactly four rays, at the price of a source that violates the NEC and has $8\pi G\rho+\Lambda\le0$ at all times.

The statement concerns the effective Einstein source. In semiclassical gravity it holds with $T_{\mu\nu}$ replaced by $\langle T_{\mu\nu}\rangle_{\mathrm{ren}}$ under two conditions: the semiclassical equation must have the Einstein null projection, and the averaged bound must hold for the total renormalized source, including geometric counterterms, in a state and metric that solve the same equations. The natural input is the self-consistent achronal ANEC conjectured by Graham and Olum \cite{GrahamOlum2007}. If it holds for the total renormalized source, no non-static Bianchi I cosmology with sections $\mathbb R^3$ is null complete in semiclassical gravity. The achronal ANEC has been proved for free fields in flat space and in two dimensions \cite{Klinkhammer1991,WaldYurtsever1991} and for general quantum field theories in Minkowski space \cite{FaulknerLeighParrikarWang2016,HartmanKunduTajdini2017}. A version averaged over transverse directions holds to second order about flat space \cite{FlanaganWald1996}, while on fixed curved backgrounds the achronal ANEC can fail for test fields \cite{UrbanOlum2010}, which is why self-consistency is required. The proofs in curved spacetime hold to first order about backgrounds that obey the NEC \cite{KontouOlum2015,Wall2010}; by Corollary~\ref{cor:ncc} of Sec.~\ref{Sec:Splitting} such a background is static if it is a null-complete Bianchi I spacetime, and $J_p$ is of second order in the departure from staticity, so none of these results constrains the class considered here.

\subsection{The anisotropy tax}
\label{Sec:Tax}

Corollary~\ref{cor:strict} says that the deficit is positive. Its density at each instant is controlled by the principal expansion rates.

\begin{proposition}[Anisotropy tax]
\label{prop:tax}
Let $n=3$, let $\kappa_{\min}\le\kappa_{\mathrm{mid}}\le\kappa_{\max}$ be the principal expansion rates at a given instant, and for an $h$-unit direction $q$ let $F(q):=\operatorname{tr}[(\mathsf P\mathsf K\mathsf P|_{q^\perp})^2]$, so that $\operatorname{tr}(B_p^2)=\omega_p^2F(q)$. Then
\begin{align}
 \min_qF&=\mu(\mathsf K):=\kappa_{\mathrm{mid}}^2+\max(\kappa_{\min},0)^2+\min(\kappa_{\max},0)^2,
 \label{eq:taxfloor}\\
 \max_qF&=\operatorname{tr}(\mathsf K^2)-\min_i\kappa_i^2,
 \label{eq:taxceiling}\\
 \langle F\rangle_{S^2}&=\frac29\theta^2+\frac7{15}\sigma^2,
 \label{eq:taxaverage}
\end{align}
where the average is uniform over the $h$-unit sphere of directions at that instant. Consequently every generator has
\begin{equation}
 J_p\ge\int_{\mathcal T}\omega_p(t)\,\mu\bigl(\mathsf K(t)\bigr)\,dt .
 \label{eq:taxintegrated}
\end{equation}
\end{proposition}

If all three rates have one sign, the floor is attained along the axis of the fastest rate, whose screen contains the two slower ones; if the signs are mixed, it is attained in the plane of the two extreme axes, at $q_{\max}^2=\kappa_{\max}/(\kappa_{\max}-\kappa_{\min})$, where the fastest expansion and the fastest contraction cancel on the screen. The ceiling is attained along the axis whose rate has the smallest magnitude. At a volume bounce, $\theta=0$, the principal rates are the shear eigenvalues and
\begin{equation}
 s_{\mathrm{mid}}^2\le F\le\sigma_B^2-s_{\mathrm{mid}}^2,
 \qquad
 \langle F\rangle_{S^2}=\frac7{15}\sigma_B^2 ,
 \label{eq:taxbounce}
\end{equation}
with $s_{\mathrm{mid}}$ the middle shear eigenvalue and $\sigma_B^2$ the squared shear at the bounce. The Hamiltonian constraint \eqref{eq:Hamiltonian} gives $8\pi G\rho+\Lambda=-\sigma_B^2/2$ there, the familiar energy cost of turning a sheared universe around; Eq.~\eqref{eq:taxbounce} is its direction-resolved counterpart for the null energy. Shear adds to the density of $J_p$ at every instant, by $\tfrac7{15}\sigma^2$ on average over the sky, and at a volume bounce, where the isotropic term $\tfrac29\theta^2$ vanishes, it is the only contribution. This puts a null-energy price on the anisotropy problem of bouncing models: on sky average, shear carried through the bounce adds $\tfrac7{15}\,\omega_p\sigma^2/(8\pi G)$ per unit time to the deficit that the sector violating the NEC must supply along a ray of frequency $\omega_p$.

Two limitations keep Proposition~\ref{prop:tax} an instantaneous statement. First, the floor vanishes when $\kappa_{\mathrm{mid}}=0$ and $\kappa_{\min}\le0\le\kappa_{\max}$, the pointwise form of saturation, and this happens at the bounce of some effective LQC solutions although no ray is saturated (Sec.~\ref{Sec:LQC}). Second, the minimum of the integrated $J_p$ over the sky is not fixed by the middle eigenvalue at the bounce, because the propagation direction $q$ rotates and the frequency $\omega_p$ redshifts differently along different rays; for the same reason the average of $J_p$ over a fixed sphere of covectors is not the time integral of \eqref{eq:taxaverage}. Section~\ref{Sec:LQC} gives what the full history adds.

\section{How much ANEC violation?}
\label{Sec:Examples}

Equation~\eqref{eq:budgetidentity} reduces the averaged null energy of a Bianchi I model to a quadrature of \eqref{eq:diagonalbulk}. We evaluate it for two classes of null-complete models with physically motivated sources, the effective bounces of LQC and turning points supported by Casimir energy \cite{ZeldovichStarobinsky1984,HerdeiroSampaio2006}. We then turn to compact spatial sections and to the metrics that attain the four-direction bound.

\subsection{Effective loop quantum cosmology}
\label{Sec:LQC}

In the effective description of LQC, a Bianchi I universe filled with a massless scalar field passes through a bounce in place of the classical singularity \cite{Chiou2007,AshtekarWilsonEwing2009,GuptSingh2012}. We use the improved effective dynamics of Ref.~\cite{AshtekarWilsonEwing2009}, summarized in Appendix~\ref{App:LQC}, with $\ell=\gamma\sqrt\Delta$ the length set by the Barbero--Immirzi parameter and the area gap, and $\rho_c=3/(8\pi G\ell^2)$ the critical density. Numerical quantum evolutions of vacuum Bianchi I states find the effective description accurate for states that bounce at volumes well above the Planck volume \cite{DienerJoeMegevandSingh2017}, and the shear through and after the bounce depends on the quantization scheme \cite{GanEtAl2026,MotaharfarSingh2026}; the values below are those of the standard improved dynamics. The effective geometries are geodesically complete \cite{Singh2012}. In the solutions used here both ends approach Kasner solutions with a scalar field (Jacobs solutions \cite{Jacobs1968}), whose exponents obey $\sum_ik_i=1$ and $k_i>-\tfrac13$, so that $\int dt/\omega_p$ diverges at both ends for every $p$. They are null-complete, non-static Bianchi I spacetimes, and the results of Sec.~\ref{Sec:Saturation} apply to their effective Einstein source $T^{\mathrm{eff}}_{\mu\nu}:=G_{\mu\nu}/(8\pi G)$. At late times $G_{00}>0$, so by Lemma~\ref{lem:energy} no direction is saturated, and by Corollary~\ref{cor:strict} the effective source violates the ANEC along every null geodesic, with deficit $J_p/(8\pi G)$. The theorems do not obstruct these bounces. They give, ray by ray, the averaged null energy that any Einsteinian source of the effective geometry must supply.

\begin{table*}[t]
\caption{Optical bulk of effective LQC Bianchi I bounces with a massless scalar field, in units of the isotropic value $J_{\mathrm{iso}}=0.739\,\omega_B/\ell$ at the same critical density. The parameters $r$ and $\eta$ are the amplitude and shape of the anisotropy, $\rho_B$ and $\sigma_B^2$ the density and squared shear at the bounce, $\sigma^2_{\max}=27/(8\ell^2)$ \cite{GuptSingh2012a}, $s_{\mathrm{mid}}$ the middle shear eigenvalue at the bounce and $H_{\max}$ the largest mean Hubble rate. The columns $J_{\min}$, $\langle J\rangle$ and $J_{\max}$ give the smallest, sky-averaged and largest $J_p$ over the null directions. The last column is the share of $J_{\min}$ accounted for by the integrated floor \eqref{eq:taxintegrated} along the minimizing ray.}
\label{tab:lqc}
\begin{ruledtabular}
\begin{tabular}{cccccccccc}
$r$ & $\eta$ & $\rho_B/\rho_c$ & $\sigma_B^2/\sigma_{\max}^2$ & $\ell s_{\mathrm{mid}}$ & $\ell H_{\max}$ & $J_{\min}/J_{\mathrm{iso}}$ & $\langle J\rangle/J_{\mathrm{iso}}$ & $J_{\max}/J_{\mathrm{iso}}$ & floor share\\
\hline\noalign{\vskip2pt}
0 & any & 1.000 & 0 & 0 & 0.500 & 1.000 & 1.000 & 1.000 & 1.00\\[3pt]
0.05 & 0 & 0.995 & 0.018 & 0 & 0.499 & 1.022 & 1.032 & 1.043 & 0.84\\
0.05 & 0.5 & 0.995 & 0.018 & 0.052 & 0.499 & 1.023 & 1.032 & 1.042 & 0.85\\
0.05 & 1 & 0.995 & 0.018 & 0.100 & 0.499 & 1.026 & 1.032 & 1.039 & 0.87\\[3pt]
0.25 & 0 & 0.886 & 0.362 & 0 & 0.485 & 1.536 & 1.712 & 1.920 & 0.42\\
0.25 & 0.5 & 0.886 & 0.362 & 0.233 & 0.485 & 1.554 & 1.712 & 1.903 & 0.46\\
0.25 & 1 & 0.886 & 0.362 & 0.451 & 0.485 & 1.602 & 1.711 & 1.857 & 0.56\\[3pt]
0.5 & 0 & 0.631 & 0.872 & 0 & 0.447 & 2.952 & 3.220 & 3.582 & 0.24\\
0.5 & 0.5 & 0.631 & 0.865 & 0.355 & 0.447 & 2.981 & 3.214 & 3.539 & 0.30\\
0.5 & 1 & 0.631 & 0.858 & 0.695 & 0.447 & 3.068 & 3.208 & 3.445 & 0.45\\[3pt]
0.8 & 0 & 0.275 & 0.949 & 0 & 0.384 & 4.952 & 5.315 & 5.563 & 0.11\\
0.8 & 0.5 & 0.277 & 0.884 & 0.317 & 0.384 & 4.848 & 5.256 & 5.868 & 0.13\\
0.8 & 1 & 0.279 & 0.817 & 0.678 & 0.384 & 4.766 & 5.194 & 5.979 & 0.20\\[3pt]
0.99 & 0 & 0.016 & 0.621 & 0 & 0.336 & 5.107 & 7.284 & 8.599 & 0.12\\
0.99 & 0.5 & 0.017 & 0.382 & 0.042 & 0.336 & 4.757 & 7.179 & 9.335 & 0.13\\
0.99 & 1 & 0.019 & 0.096 & 0.232 & 0.336 & 4.366 & 7.002 & 10.086 & 0.18\\
\end{tabular}
\end{ruledtabular}
\end{table*}

\begin{figure*}[t]
 \centering
 \includegraphics[width=\textwidth]{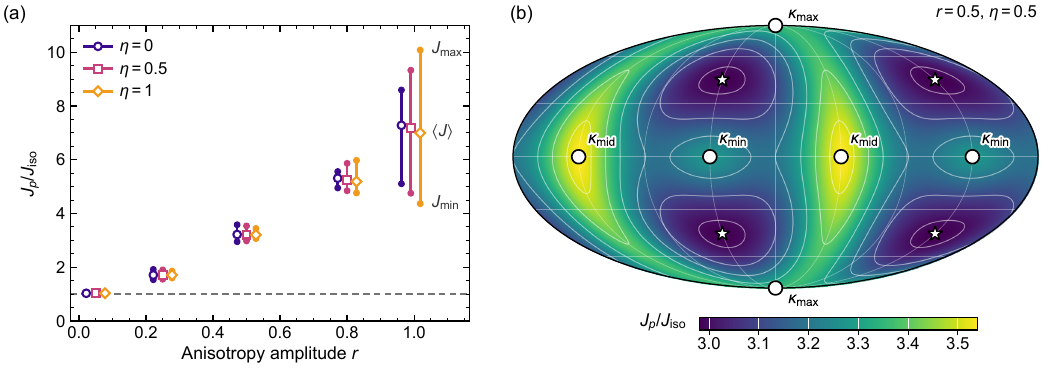}
 \caption{Null-energy deficit of effective LQC Bianchi I bounces at fixed critical density. (a) Smallest, sky-averaged and largest optical bulk $J_p$ over the null directions, in units of the isotropic value $J_{\mathrm{iso}}=0.739\,\omega_B/\ell$, against the anisotropy amplitude $r$ for the shapes $\eta=0$, $0.5$ and $1$, offset horizontally for clarity; $\omega_B$ is the frequency at the bounce. Bars span $J_{\min}$ to $J_{\max}$, and open symbols mark the average $\langle J\rangle$ over directions distributed uniformly at the bounce; $\eta$ and $-\eta$ give the same values. The dashed line marks the isotropic value. (b) $J_p/J_{\mathrm{iso}}$ over the directions at the bounce, in Mollweide projection, for $(r,\eta)=(0.5,0.5)$, with contours in steps of $0.1$. Circles mark the principal axes of the bounce, labeled by their expansion rates $\kappa_{\min}<\kappa_{\mathrm{mid}}<\kappa_{\max}$. The minimum, $2.98$ (stars), lies in the plane of the fastest contracting and fastest expanding axes, and the maximum, $3.54$, on the middle axis. In Einstein gravity the averaged null energy along each ray is $-J_p/(8\pi G)$.}
 \label{fig:lqc}
\end{figure*}

\paragraph*{Isotropic bounce.}
The effective Friedmann equation $H^2=\tfrac{8\pi G}{3}\rho(1-\rho/\rho_c)$ with $\rho\propto a^{-6}$ gives
\begin{equation}
 a=a_B(1+\alpha t^2)^{1/6},
 \qquad
 H=\frac{\alpha t}{3(1+\alpha t^2)},
 \label{eq:LQCiso}
\end{equation}
with $\alpha:=24\pi G\rho_c=9/\ell^2$, so the Hubble rate peaks at $H_{\max}=\sqrt\alpha/6=1/(2\ell)$, at $t=\pm t_b$ with $t_b:=\alpha^{-1/2}=\ell/3$. Here $B_p=\omega_pH\mathsf 1$ and $\omega_p=\omega_B(1+\alpha t^2)^{-1/6}$, where $\omega_B$ is the frequency at the bounce. The resulting integral was evaluated in closed form by Li and Zhu, Eq.~(27) of Ref.~\cite{LiZhu2009}, who found the ANEC violation of isotropic LQC; in terms of the frequency at the bounce,
\begin{equation}
 \begin{aligned}
 J_p&=\frac{2\omega_B\sqrt\alpha}{9}\int_{-\infty}^{\infty}\frac{x^2\,dx}{(1+x^2)^{13/6}}\\
    &=\frac{2\,\Gamma(\tfrac23)\Gamma(\tfrac56)}{7\sqrt\pi}\,\omega_B\sqrt{24\pi G\rho_c}
     =1.4783\,\omega_BH_{\max}
 \end{aligned}
 \label{eq:LQCclosed}
\end{equation}
along every ray. We denote this value by $J_{\mathrm{iso}}$; at a given critical density, $J_{\mathrm{iso}}=0.739\,\omega_B/\ell$. The direct integral of $R_{\mu\nu}k^\mu k^\nu=-2\omega_p^2\dot H$ gives $-J_p$. The scalar field itself obeys the NEC. Its null energy integrates to $8\pi G\int T^{\phi}_{\mu\nu}k^\mu k^\nu d\lambda=+7J_p$, and the quantum-geometry correction contributes $-8J_p$. With $\rho_c\approx0.41\,\rho_{\mathrm P}$ \cite{AshtekarWilsonEwing2009}, $H_{\max}\approx0.93\,m_{\mathrm P}$ and the deficit is $J_p/(8\pi G)\approx0.055\,\omega_Bm_{\mathrm P}^3$.

\paragraph*{Triaxial bounces.}
Anisotropic solutions are labeled by conserved differences of the connection variables (Appendix~\ref{App:LQC}). We scan a two-parameter family: an amplitude $r\in[0,1)$, with $r^2=\sigma^2/(\sigma^2+16\pi G\rho)$ the shear parameter $\Sigma^2$ \cite{WainwrightEllis1997} of the two asymptotic ends, the same on both sides of the bounce, and a shape $\eta\in[-1,1]$ of the shear spectrum at the bounce, with $\eta=0$ the spectrum $(-s,0,s)$ whose middle eigenvalue vanishes and $\eta=\pm1$ the two axisymmetric spectra. The amplitude is a property of the ends, not of the bounce. The squared shear at the bounce, $\sigma_B^2$, is not monotone in $r$ and decreases again toward the vacuum limit $r\to1$. All covectors have the same frequency $\omega_B$ at the bounce, and sky averages are uniform over the directions at the bounce. We compare the solutions at fixed critical density and quote $J_p$ in units of $J_{\mathrm{iso}}$, because the largest mean Hubble rate $H_{\max}:=\max_t|\theta|/3$ is not a fixed unit. It falls from $1/(2\ell)$ at $r=0$ to $1/(3\ell)$ as $r\to1$. Table~\ref{tab:lqc} and Fig.~\ref{fig:lqc} summarize a grid of six amplitudes and five shapes with 2304 directions each; shapes of opposite sign are related by time reversal.

In every direction of every background $J_p>0$ (the independently computed Ricci integral equals $-J_p$ to a relative accuracy of $7\times10^{-14}$), and anisotropy raises the deficit steeply. The smallest value over the sky grows to $1.5$--$1.6\,J_{\mathrm{iso}}$ at $r=0.25$, $3.0$--$3.1\,J_{\mathrm{iso}}$ at $r=0.5$ and $4.8$--$5.0\,J_{\mathrm{iso}}$ at $r=0.8$, and then levels off, at $4.4$--$5.1\,J_{\mathrm{iso}}$ for $r=0.99$, so the cheapest ray of a strongly anisotropic bounce pays four to five times the isotropic value. The sky average keeps growing, to $7.0$--$7.3\,J_{\mathrm{iso}}$ at $r=0.99$, where the largest value reaches $10.1\,J_{\mathrm{iso}}$; in the vacuum limit of the $\eta=0$ family $J_{\min}\to5.08\,J_{\mathrm{iso}}$ and $J_{\max}\to8.87\,J_{\mathrm{iso}}$. Integrating the instantaneous average \eqref{eq:taxaverage} with the sky-averaged frequency $\langle\omega_p\rangle$ falls short of the sky mean $\langle J\rangle$ by $0.9\%$ at $r=0.05$, $12\%$ at $r=0.25$, $21\%$ at $r=0.5$ and $30\%$ at $r=0.99$. The difference comes from the direction dependence of the redshift and from the change of angular measure as the directions rotate; its sign is specific to this family.

The cheapest ray changes with the amplitude. In the $\eta=0$ family it lies, for $r\le0.6$, in the plane of the extreme principal axes of the bounce, on the lines where the floor \eqref{eq:taxbounce} is attained. It moves away between $r\approx0.68$ and $0.85$, and for $r\ge0.9$ it is the middle axis. At large anisotropy the Kasner ends, not the bounce instant, decide which ray is cheapest. The middle shear eigenvalue likewise does not control $J_{\min}$. It vanishes throughout the $\eta=0$ family, so that the floor \eqref{eq:taxbounce} is zero and two lines have a saturated instant at the bounce, yet $J_{\min}>J_{\mathrm{iso}}$. Across shapes, $J_{\min}$ rises with $|s_{\mathrm{mid}}|$ at $r=0.25$ and $0.5$ and falls with it at $r=0.8$ and $0.99$ (Table~\ref{tab:lqc}). The integrated floor \eqref{eq:taxintegrated} is a strict lower bound; it accounts for $84$--$87\%$ of $J_{\min}$ at $r=0.05$ and is a factor of five to nine below it at $r\ge0.8$.

The two ends become vacuum Kasner solutions only as $r\to1$, and the change of exponents across the bounce is the Kasner transition of Refs.~\cite{GuptSingh2012,WilsonEwing2018}. Shear and scalar density both scale as $\mathcal V^{-2}$ at the ends, so the shear never becomes negligible and the tails of $J_p$ depend on direction. The shear at the bounce depends on the shape as well as on the density. Near the vacuum limit the $\eta=0$ solutions carry more shear at a given bounce density than the trend reported in Ref.~\cite{McNamaraSainiSingh2023} (Table~\ref{tab:lqc}).

\subsection{A Casimir turning point on \texorpdfstring{$T^3$}{T3}}
\label{Sec:CasimirTurn}

Consider a flat FLRW metric on a cubic torus, in comoving coordinates of unit period, with scale factor
\begin{equation}
 a(t)=a_b\sqrt{\cosh 2H_\infty t},
 \qquad
 a_b^4=\frac{c_3}{3\bar M_{\mathrm P}^2H_\infty^2},
 \label{eq:casimirturn}
\end{equation}
so that $a$ is the physical edge of the torus. It is sourced by the Casimir vacuum $\rho_C=-c_3/a^4=3P_C$ of massless fields together with a constant potential $V=3\bar M_{\mathrm P}^2H_\infty^2$; here $c_3=0.8375$ per periodic real scalar, without its zero mode, and $0.3914$ per antiperiodic Weyl fermion (Appendix~\ref{App:Casimir}). This is the solution of Zel'dovich and Starobinsky \cite{ZeldovichStarobinsky1984}, rederived in Ref.~\cite{GuthVilenkin2025}; in Ref.~\cite{Zhang2026Casimir} one of us evaluates $c_3$ for the Standard Model and reads the turn as a beginning. The Friedmann equation $H^2=(V+\rho_C)/(3\bar M_{\mathrm P}^2)=H_\infty^2[1-(a_b/a)^4]$ gives $H=H_\infty\tanh2H_\infty t$, so $H_\infty$ is the asymptotic Hubble rate. The spacetime is null complete, since $\int a\,dt$ diverges at both ends, and has no saturated line, since $\mathsf K=H\,\mathsf 1$ has rank three for $t\neq0$. At $t=0$ every direction is instantaneously saturated and $G_{00}=0$, as Lemma~\ref{lem:energy} requires. The optical bulk is the same in every direction,
\begin{equation}
 J_p=-I_p=\frac{2\sqrt\pi\,\Gamma(5/4)}{\Gamma(7/4)}\,\omega_bH_\infty\approx3.4961\,\omega_bH_\infty,
 \label{eq:casimirdeficit}
\end{equation}
where $\omega_b=|p|/a_b$ is the frequency at the turn. The Casimir vacuum pays it through $\rho_C+P_C=\tfrac43\rho_C<0$; the potential contributes nothing.

On a torus a timelike curve can join two points of a light ray by going around a compact direction, so complete light rays need not be achronal (Sec.~\ref{Sec:Toroidal}). During the whole history light crosses the comoving cell $N_{\mathrm{cross}}=\int dt/a=c_\times/(H_\infty a_b)$ times, with $c_\times=\sqrt\pi\,\Gamma(\tfrac14)/[2\Gamma(\tfrac34)]\approx2.622$. Since the metric is conformally flat, the generator with unit direction $n$ is achronal if and only if $N_{\mathrm{cross}}\le1/(2\max_i|n_i|)$ (Appendix~\ref{App:Toroidal}), so no null geodesic is achronal for $H_\infty a_b<2c_\times/\sqrt3\approx3.03$ and every one is achronal for $H_\infty a_b\ge2c_\times\approx5.24$. Casimir support fixes $H_\infty a_b=(c_3/3)^{1/4}(H_\infty/\bar M_{\mathrm P})^{1/2}$, so a turn with $H_\infty\ll(3/c_3)^{1/2}\bar M_{\mathrm P}$ has $H_\infty a_b\ll1$; for the parameters of Ref.~\cite{Zhang2026Casimir}, which have $c_3\approx30$, one finds $H_\infty a_b\approx0.0044$ and $N_{\mathrm{cross}}\approx600$, and an achronal generator would need $H_\infty\gtrsim2.9\,\bar M_{\mathrm P}$. The ANEC violation \eqref{eq:casimirdeficit} is therefore invisible to the achronal ANEC.

\subsection{Compact sections}
\label{Sec:Toroidal}

Corollary~\ref{cor:averagedrigidity} uses $\Sigma=\mathbb R^n$ only to make every complete generator achronal. Compact directions are the standard example of why the achronal qualifier is needed: Casimir stresses violate the ANEC along null geodesics that wind around a compact direction, and those geodesics are not achronal \cite{GrahamOlum2007,FewsterOlumPfenning2007}. Let $\Sigma=\mathbb R^3/\Gamma$, where $\Gamma$ is a lattice of comoving translations of rank $c$, so that $\Sigma\cong\mathbb R^{3-c}\times T^c$, and call a covector winding if $p(L)\neq0$ for some $L\in\Gamma$. Two-sided complete generators of covectors that do not wind are null lines, because the optical function \eqref{eq:opticalfunction} is then defined on the quotient and decreases along every future-directed timelike curve (Lemma~\ref{lem:nulllines}). A winding generator has a timelike shortcut, the generator itself displaced steadily by one lattice vector, as soon as $\int d\lambda/h(L,L)>1/(2|p(L)|)$ over a segment of it, in particular whenever it is complete in a direction along which $h(L,L)$ stays bounded (Lemma~\ref{lem:winding} in Appendix~\ref{App:Toroidal}). On a torus whose metric stays bounded no complete null geodesic is achronal, so achronal averaged conditions are empty there, although the ANEC itself can fail along these geodesics \cite{KontouSanders2020,FewsterOlumPfenning2007}. On an expanding torus the integral can stay bounded, and complete achronal winding geodesics then exist, as in Sec.~\ref{Sec:CasimirTurn} for large $H_\infty a_b$.

Together with the results of Sec.~\ref{Sec:Saturation} these facts give the following statement (Proposition~\ref{prop:toroidal} in Appendix~\ref{App:Toroidal}). For toroidal spatial sections $\mathbb R^{3-c}\times T^c$ obtained by identifying comoving translations, the achronal ANEC excludes every null-complete, non-static Bianchi I spacetime with fewer than two compact directions; with two, it forces the torus to be exactly static and the total energy density, including the cosmological constant, to be nonpositive, and with three it constrains only universes in which some light rays never circumnavigate the torus. Compact directions cannot remove the ANEC violation, which by Corollary~\ref{cor:strict} persists along every complete null geodesic outside at most two lines; they can only make the violating light rays wind around the torus, where they need not be achronal.

With fewer than two compact directions, three distinct non-winding lines exist and the proof of Corollary~\ref{cor:averagedrigidity} applies unchanged. With two, only the non-compact line $[dz]$ is achronal. Saturating it freezes the torus, shape included, and leaves the evolution of the non-compact direction and of a tilt $b(t)$ that shears the torus against it, as in the normal form \eqref{eq:normalform}. For constant $b$ the spacetime is a two-dimensional cosmology $-dt^2+a^2dz^2$ times a static torus. The null curvature of a two-dimensional metric vanishes identically, so saturation of $[dz]$ costs nothing, and the source must have $\rho+P_z=0$, as the Casimir vacuum of the torus has by boost invariance in the $(t,z)$ plane. A varying tilt is also allowed and carries $8\pi G\rho+\Lambda=-|v|^2\le0$, negative wherever $\dot b\neq0$.

With three compact directions every covector is winding, but a light ray has a timelike shortcut only if the comoving distance light travels during the whole history is comparable with the size of the torus. The achronal ANEC excludes tori that are large compared with this distance, as it excludes $\mathbb R^3$, and is empty for tori that are small compared with it, as at the Casimir turning point above. The restriction to translations is needed, because a quotient of $\mathbb R^3$ by a screw motion carries a null-complete, non-static Bianchi I metric with no null line at all (Appendix~\ref{App:Toroidal}).

\paragraph*{A frozen two-torus.}
The two-torus branch is realized by the universal cover of $dS_2$ times a flat square torus,
\begin{equation}
 ds^2=-dt^2+\cosh^2(\mu t)\,dz^2+a_T^2(dx^2+dy^2),
 \label{eq:dStorus}
\end{equation}
with $x$ and $y$ of unit period and constant side $a_T$. Its Einstein tensor has $G^t{}_t=G^z{}_z=0$ and $G^x{}_x=G^y{}_y=-\mu^2$. At leading semiclassical order it is supported by the Casimir stress of massless fields with negative Casimir energy, $\rho_C=-P_z=P_x=P_y<0$, together with a cosmological constant, $\rho_\Lambda=|\rho_C|$ and $\mu^2=16\pi G|\rho_C|$, with corrections of relative order $\mu^2a_T^2\sim N\ell_{\mathrm P}^2/a_T^2$ for $N$ fields; a periodic real scalar and an antiperiodic Weyl fermion each give $\rho_C\approx-0.305/a_T^4$ (Appendix~\ref{App:Casimir}). It is the two-dimensional de Sitter vacuum of the torus-reduced theory of Ref.~\cite[App.~B.3]{ArkaniHamedDubovskyNicolisVilladoro2007}. Like all such vacua it sits at a maximum of the moduli potential and is unstable: the torus size grows as $\sinh\mu t$ under homogeneous perturbations. The solution is null and timelike complete, $[dz]$ is its only saturated line, and $G_{00}\equiv0$. Every winding direction has $I_p^{\mathrm{ind}}=-\infty$ and is not achronal. The two-torus loophole is therefore not empty. This null-complete, non-static Bianchi I spacetime obeys the achronal ANEC, with one exempt line and a static torus. Appendix~\ref{App:Toroidal} gives the chronology of its winding rays, the tilted solutions on rectangular tori, the shape instabilities and the variant with $z$ periodic.

\subsection{Directions that escape the budget}
\label{Sec:SharpExamples}

The four-direction bound of Proposition~\ref{prop:classification} is attained by the periodic model of Fig.~\ref{fig:sharpmodel}. Let $\mathcal T=\mathbb R$ and
\begin{equation}
 \begin{aligned}
 a_1^2&=X:=X_0+\varepsilon\cos t,\\
 a_2^2&=Y:=Y_0-\varepsilon\cos t,\\
 a_3&=c,
 \end{aligned}
    \label{eq:sharpperiodicmodel}
\end{equation}
with constants $X_0,Y_0>|\varepsilon|>0$ and $c>0$. The metric is bounded above and below, so the spacetime is null and timelike complete, and all orthonormal curvature components are bounded. In the form \eqref{eq:twolinefamily}, $p=(dx^1+dx^2)/\sqrt2$, $r=(dx^1-dx^2)/\sqrt2$ and $\Phi=\varepsilon\cos t$, and the saturated rays are
\begin{equation}
    (p_1,p_2,p_3)\propto(\pm1,\pm1,0).
    \label{eq:sharpmomenta}
\end{equation}
On these rays every cutoff integral vanishes. For every other direction the cosmic-time density of the optical bulk in \eqref{eq:Jcosmictime} is periodic with positive mean [Fig.~\ref{fig:sharpmodel}(b) and (c)], so both tails of $J_p$ grow linearly while $\Theta_p$ stays bounded, and $\underline I_p=I_p^{\mathrm{ind}}=-\infty$. Burwig and Easson found the same divergence for periodic flat and open FLRW models \cite{BurwigEassonAffine2026}.

\begin{figure}[t]
    \centering
    \includegraphics[width=\columnwidth]{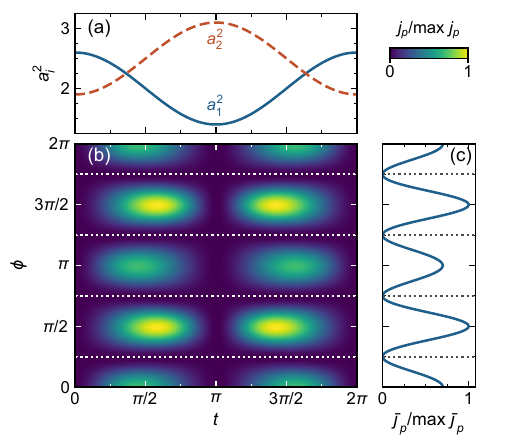}
    \caption{Periodic model \eqref{eq:sharpperiodicmodel} with $X_0=2$, $Y_0=2.5$, $\varepsilon=0.6$. (a) Squared scale factors $a_1^2$ and $a_2^2$ over one period; their sum is constant. (b) Optical-bulk density $j_p=\omega_p\operatorname{tr}[(\mathsf P\mathsf K\mathsf P|_{\ker p})^2]$, the integrand of Eq.~\eqref{eq:Jcosmictime}, for coordinate-unit covectors $p=(\cos\phi,\sin\phi,0)$, normalized to its maximum. It vanishes at the instants $t=0,\pi,2\pi$, where the shear vanishes, and at all times on the four saturated rays $p_1=\pm p_2$ (dotted). (c) Its period average $\bar\jmath_p$, normalized to its maximum. It is positive except on the saturated rays, so every other null direction has $I_p^{\mathrm{ind}}=-\infty$.}
    \label{fig:sharpmodel}
\end{figure}

For vanishing cosmological constant the effective Einstein source has $8\pi G\rho=-\varepsilon^2\sin^2t/(4XY)\le0$, the cost of saturation identified in Lemma~\ref{lem:energy}, and saturation along the rays \eqref{eq:sharpmomenta} forces the principal null combinations $\rho+P_1$ and $\rho+P_2$ to have opposite signs except at isolated times. The NEC therefore fails at all times outside a discrete set, and condition~(ii) of Theorem~\ref{thm:main} fails, as it must (Appendix~\ref{App:Pulses}).

The same two-line class contains pulses of compact support with a static exterior, $h=\operatorname{diag}(X_0+f,Y_0-f,c^2)$ with $f\in C_c^\infty(\mathbb R)$. For them $I_p^{\mathrm{ind}}=-J_p^{\mathrm{tot}}$ is finite, vanishes exactly on the rays \eqref{eq:sharpmomenta} and is strictly negative elsewhere, and stretching the pulse in time by a factor $T$ divides $J_p^{\mathrm{tot}}$ by $T$, so a small deficit constrains the amplitude only when the duration is controlled. A single saturated line allows evolution that is non-diagonal in every fixed basis, as in the normal form \eqref{eq:normalform} with $a=1$ and a bounded, nonconstant tilt $b(t)$. Both are worked out in Appendix~\ref{App:Pulses}.

\section{Pointwise null energy and incompleteness}
\label{Sec:Splitting}

The averaged statements have a pointwise counterpart that needs only one complete null geodesic and curvature conditions in two null directions. It shows that under the NEC no null geodesic of a non-static Bianchi I spacetime is complete, and it implies the NEC violation that accompanies every saturated direction (Sec.~\ref{Sec:Saturation}).

For a nonzero covector $p$ and $t\in\mathcal T$, let $\{E_A(t)\}_{A=1}^{n-1}$ be an $h(t)$-orthonormal basis of $\ker p$, and let $\ell_A=\partial_t+E_A$ be the corresponding null vectors. Because $R_{0i}=0$, the screen average
\begin{equation}
 \begin{aligned}
 \mathcal N_p(t)
   &:=\frac{1}{n-1}\sum_{A=1}^{n-1}
      R_{\mu\nu}\ell_A^\mu\ell_A^\nu\\
   &=R_{00}+\frac{1}{n-1}
      \operatorname{tr}_{(\ker p,\,h)}\mathrm{Ric}
 \end{aligned}
 \label{eq:screenaverage}
\end{equation}
does not depend on the choice of basis. It is the null curvature averaged over the null directions whose spatial part lies in the screen of $p$.

\begin{theorem}[Staticity from one complete light ray]
\label{thm:main}
Let \eqref{eq:BImetric} be a smooth Bianchi I metric on $\mathcal T\times\mathbb R^n$, $n\ge2$. Suppose a null geodesic $\gamma$ with conserved covector $p$ is complete in both affine directions, and
\begin{enumerate}
\item[(i)] $R_{\mu\nu}k^\mu k^\nu\ge0$ along $\gamma$;
\item[(ii)] for every $t\in\mathcal T$ there is an $h(t)$-unit vector $e\in\ker p$ with $R_{\mu\nu}\ell^\mu\ell^\nu\ge0$ for $\ell=\partial_t+e$.
\end{enumerate}
Then $\mathcal T=\mathbb R$ and $h(t)=h_0$ is constant. The Cartesian cover is Minkowski space after a constant linear change of spatial coordinates, and the same holds for the universal cover of any compatible flat spatial quotient to which the metric descends.
\end{theorem}

The hypotheses involve only two null directions at each time, the direction of $\gamma$ and one transverse direction $\partial_t+e$; this is the sense in which the rigidity is directional. The mechanism has two steps. Along the complete ray, condition (i) makes $(\mathcal V\omega_p)^{1/(n-1)}$ a positive concave function of the affine parameter on the whole real line, hence constant, so the light bundle is saturated and the metric has the normal form \eqref{eq:normalform}. In that form the transverse null curvature is
\begin{equation}
 R_{\mu\nu}\ell^\mu\ell^\nu
    =-\frac{\ddot a}{a}-2|v|^2-2(v\cdot e)^2
 \label{eq:concavity}
\end{equation}
for $\ell=\partial_t+e$ with $e$ a unit vector in $\ker p$, so condition (ii) makes $a$ concave in cosmic time, which completeness again turns into a constant. Condition (ii) holds in particular when $\mathcal N_p(t)\ge0$. Neither condition can be dropped. The periodic model of Sec.~\ref{Sec:SharpExamples} has complete saturated generators with $R_{\mu\nu}k^\mu k^\nu=0$ and violates (ii). For $h=\operatorname{diag}(1,e^{-t^2},e^{-t^2})$ on $\mathbb R$ and $p=dx^1$ one has $\omega_p=1$, so the generators are complete, while $\mathcal N_p=1$ and $R_{\mu\nu}k^\mu k^\nu=2-2t^2$; here (ii) holds and (i) fails.

\begin{corollary}[Null energy condition]
\label{cor:ncc}
If the NCC $R_{\mu\nu}\ell^\mu\ell^\nu\ge0$ holds for all null $\ell$, every null geodesic of a non-static Bianchi I spacetime is incomplete in at least one affine direction. In Einstein gravity with arbitrary cosmological constant,
\begin{equation}
 \begin{aligned}
 R_{\mu\nu}k^\mu k^\nu
   &=8\pi G\,T_{\mu\nu}k^\mu k^\nu,\\
 \mathcal N_p&=8\pi G\,(\rho+P_\perp),\\
 P_\perp&:=\frac{1}{n-1}
              \operatorname{tr}_{(\ker p,\,h)}T,
 \end{aligned}
 \label{eq:EinsteinNull}
\end{equation}
so the conclusion of Theorem~\ref{thm:main} holds whenever $T_{\mu\nu}k^\mu k^\nu\ge0$ along $\gamma$ and, at each time, $\rho+T(e,e)\ge0$ for one $h$-unit $e\in\ker p$; this follows from $\rho+P_\perp\ge0$ and in particular from the NEC.
\end{corollary}

That the NEC obstructs a complete history is expected, and for a spacetime that is null complete as a whole it also follows from null splitting together with homogeneity (Remark~\ref{rem:timelike} in Appendix~\ref{App:Rigidity}). What Corollary~\ref{cor:ncc} adds is that each null geodesic is incomplete by itself, in every non-static history. With $\mathcal N_{\mathrm{all}}=R_{00}+h^{ij}R_{ij}/n$, the instantaneous average of the null curvature over spatial directions, Eq.~\eqref{eq:RicciGeneral} gives $\dot\theta=-\tfrac{n}{n-1}(\sigma^2+\mathcal N_{\mathrm{all}})$, so the NCC makes the volume expansion nonincreasing and excludes a volume bounce, but it does not determine the directional affine lengths. Minkowski space and expanding flat de Sitter space both have $\mathcal N_{\mathrm{all}}=0$ and differ in past null completeness.

Corollary~\ref{cor:ncc} covers the histories that the volume argument does not reach, such as eternal expansion or contraction from a static past, histories in which individual scale factors bounce, and non-diagonal metrics. It also covers those with a direction that contracts at all times, which lie outside the hypotheses of the theorem of Ref.~\cite{GarciaSaenzHuaSherif2026}. Any past-complete null geodesic of a smoothed version of the counterexample given there is future incomplete. In Bianchi I the NCC thus forces null incompleteness despite the shear, in contrast with the kinematic argument analyzed in Ref.~\cite{KinneyBGVNCC2026}. Corollary~\ref{cor:ncc} also contains the Bianchi I analogue of the flat FLRW conjecture of Ref.~\cite{EassonLesnefsky2025}. For flat FLRW metrics with scale factor bounded above and away from zero, staticity under the NCC and null completeness follows from Theorem~2 of Ref.~\cite{FolkestadHernandezCuenca2021}. The effective Bianchi I geometries of LQC bounce and are geodesically complete \cite{Chiou2007,AshtekarWilsonEwing2009,Singh2012}; the corollary says that their effective Einstein source violates the NEC, and Sec.~\ref{Sec:LQC} says by how much on average.

Incompleteness of the Bianchi I region need not signal a singularity. The flat slicing of de Sitter space, $h=e^{2Ht}\delta$, satisfies the NCC, is past null incomplete and extends to global de Sitter space, and the Kasner metric $h=\operatorname{diag}(t^2,1,\ldots,1)$ on $t>0$ is a Milne wedge times $\mathbb R^{n-1}$ inside Minkowski space \cite{NomuraYoshida2021,YoshidaQuintin2018,GeshnizjaniLingQuintin2023,GallowayLingSbierski2018,EllisKing1974}.

Theorem~\ref{thm:main} does not say which end of a geodesic is incomplete. The same congruences give a one-sided criterion in terms of the principal expansion rates $\kappa_1\le\cdots\le\kappa_n$: under the NCC, if $\kappa_1+\cdots+\kappa_{n-1}>0$ at some time $t_0$, every null geodesic is past incomplete, with affine length at most $(n-1)/\Theta_p(t_0)$ to the past of $t_0$, and if $\kappa_2+\cdots+\kappa_n<0$ every null geodesic is future incomplete (Proposition~\ref{prop:onesided} and Corollary~\ref{cor:eigen} in Appendix~\ref{App:Rigidity}). If all principal directions expand at one time the first hypothesis holds, so this extends the past incompleteness of the axial null geodesics of diagonal metrics \cite{GarciaSaenzHuaSherif2026} to every null geodesic and to non-diagonal metrics, with an affine bound.

Both hypotheses of Theorem~\ref{thm:main} can be replaced by averages. Condition~(i) can be replaced by any of the averaged conditions of Sec.~\ref{Sec:Averaged} (Corollary~\ref{cor:bianchiaveraged}). When cosmic time is complete, condition~(ii) can be replaced by an average of $\mathcal N_p$ along the comoving worldlines (Corollary~\ref{cor:fullyaveraged} in Appendix~\ref{App:AveragedProofs}); since $\mathcal N_p=8\pi G(\rho+P_\perp)$ in Einstein gravity, this is an averaged energy condition along a timelike worldline, the setting in which null-projected quantum inequalities are formulated \cite{FewsterRoman2003,Fewster2000}.

\section{Discussion}
\label{Sec:Discussion}

In Einstein gravity a null-complete, non-static Bianchi I universe has averaged null energy at most $-J_p/(8\pi G)$ along each complete light ray, with equality whenever the integral converges. At most four null directions are exempt, and none if the total energy density is ever positive. For a model builder the identity is a two-step recipe. Equation~\eqref{eq:diagonalbulk} gives the density of $J_p$ for any diagonal background from its directional Hubble rates, and by Eq.~\eqref{eq:modelbuilder} every Einsteinian source of that background must satisfy $8\pi G\int\omega_p(\rho+P_{ij}q^iq^j)\,dt=-J_p$ along each complete ray on which the integral converges. Both steps are quadratures, so the budget can be evaluated for any bounce mechanism and compared with the shear that its contracting phase delivers to the bounce. For the LQC and Casimir models of Sec.~\ref{Sec:Examples} the result is $J_p=1.48\,\omega_BH_{\max}$ at the isotropic LQC bounce, four to five times as much on the cheapest ray and up to ten times as much on the most expensive ray of a strongly anisotropic bounce at the same critical density, and $3.50\,\omega_bH_\infty$ at the Casimir turn.

\paragraph*{Trace and null energy.}
The trace of Einstein's equations constrains a different part of the source. With $T=g^{\mu\nu}T_{\mu\nu}$ including the cosmological term, Eq.~\eqref{eq:RicciGeneral} gives $R=2\dot\theta+\tfrac43\theta^2+\sigma^2=-8\pi G\,T$ for $n=3$. At an instant with $\theta=0$ the Hamiltonian constraint gives $8\pi G\rho+\Lambda=-\sigma^2/2$, and $\dot\theta>0$ requires $8\pi G\,T<-\sigma^2$. A turning point of an anisotropic flat universe therefore has two requirements. The deficit $J_p/(8\pi G)$ must come from a sector that violates the NEC; a potential contributes nothing to it. The negative trace at the turn must come from non-conformal matter, most cheaply from a potential, since conformal fields contribute to the trace only through the anomaly. Reference~\cite{Zhang2026Casimir} argues that conformal matter alone cannot turn an anisotropic universe around. On a torus, the Casimir vacuum of conformal fields is traceless and has negative null energy, carried by light rays that wind, so it supplies the deficit and not the negative trace. Casimir energy together with a potential, as in Sec.~\ref{Sec:CasimirTurn}, is in this sense a minimal source.

The saturating metrics of Sec.~\ref{Sec:SharpExamples} need anisotropic stress as well. A homogeneous canonical or phantom scalar, a $P\bigl(\varphi,(\partial\varphi)^2\bigr)$ field and a cubic Galileon have isotropic pressure, and in a regular $F(\varphi)R$ theory the shear obeys $d(\mathcal VF\sigma^i{}_j)/dt=0$ \cite{Kobayashi2019}, so a homogeneous scalar can neither switch on the shear from a static exterior nor carry the periodic model through its instants of vanishing shear. Vector fields, $p$-forms and fluxes \cite{WatanabeKannoSoda2009,NormannEtAl2018,KoivistoMota2008,KamalMahanta2026} are possible anisotropic sectors, but the total source must still violate the NEC, for instance through an isotropic sector of the kind studied in Refs.~\cite{CreminelliLutyNicolisSenatore2006,EassonSawickiVikman2011,IjjasSteinhardt2016}, whose stability is a separate problem \cite{DobreEtAl2018,Kobayashi2016,LibanovMironovRubakov2016,CreminelliPirtskhalavaSantoniTrincherini2016}. We do not construct a matter model for those metrics.

\paragraph*{Finite episodes.}
For an anisotropic episode of duration $T$ between static epochs every term of the identity is finite, and the deficits along three directions bound the excursion of the metric (Proposition~\ref{prop:quantitative} in Appendix~\ref{App:Pulses}). In coordinates in which the metric before the episode is the identity matrix, the net metric change satisfies $\|h(+\infty)-h(-\infty)\|_{\mathrm F}^2\le4TM^{5/2}\sum_a\delta_a/c_{\mathcal P}$ [Eq.~\eqref{eq:stepbudget}], where $\delta_a/(8\pi G)$ bound the averaged null-energy deficits along the three directions, for rays of unit frequency before the episode, $M$ bounds the metric during the episode, $c_{\mathcal P}$ depends only on the directions and $\|\cdot\|_{\mathrm F}$ is the Frobenius norm. If such an episode is embedded in an expanding cosmology, it could leave a fossil anisotropy in a primordial spectrum \cite{DaiJeongKamionkowski2013} relevant to statistical-isotropy searches \cite{AckermanCarrollWise2007,KimKomatsu2013,Planck2018Isotropy,SaadehEtAl2016}. The net change is a constant linear map of the comoving coordinates and is observable only relative to a reference fixed before the episode; turning it into a quadrupole requires the scalar and tensor transfer functions through the background \cite{PereiraPitrouUzan2007,GumrukcuogluContaldiPeloso2007}, which we do not compute. The bound needs static ends. Between inflating ends, as at the turning point of Sec.~\ref{Sec:CasimirTurn}, it does not bound the net change of shape.

\paragraph*{Semiclassical reading.}
The results constrain the effective Einstein source. Under the self-consistent achronal ANEC of Graham and Olum they become a conditional statement about semiclassical gravity (Sec.~\ref{Sec:AveragedRigidity}), and compact sections are the loophole that the achronal qualifier anticipates. On toroidal sections the loophole needs at least two compact directions and a torus that light circumnavigates during the history (Proposition~\ref{prop:toroidal}). Casimir support provides the second automatically, because a Casimir-supported turn with $H_\infty\ll\bar M_{\mathrm P}$ has $H_\infty a_b\ll1$ (Sec.~\ref{Sec:CasimirTurn}).

\paragraph*{Limits.}
The results rest on homogeneity and on flat slices. The Abelian translation group fixes the screen and rules out caustics. With curved slices the conclusion can fail: closed de Sitter space is a null-complete Bianchi IX model with vanishing null curvature. A left-invariant covector is closed exactly when it annihilates the derived algebra. Such covectors exist for Bianchi types II--VII but not for types VIII and IX \cite{EllisMacCallum1969}, and for these covectors the optical function and the caustic-free congruence persist, so that Proposition~\ref{thm:path} applies along their rays. Whether Theorem~\ref{thm:main} and the count of saturated directions extend to types II--VII, in parallel with the curvature classification of FLRW models \cite{BurwigEassonPRD2026,BurwigEassonAffine2026}, is left to a follow-up. Under inhomogeneous perturbations the screen moves and caustics can form, and control of the optical bulk on a complete ray is open. The LQC values assume that the effective equations of the standard improved dynamics hold through the bounce (Sec.~\ref{Sec:LQC}). The Casimir examples are leading-order semiclassical solutions, and the frozen two-torus is unstable. Tensor modes, whose mixing through an anisotropic episode requires the coupled perturbation equations \cite{PereiraPitrouUzan2007}, and the transfer of a fossil anisotropy to primordial spectra are the natural next steps toward observables.

\begin{acknowledgments}
A.Z.\ was supported by a graduate scholarship and research support from McMaster University. Research at Perimeter Institute is supported in part by the Government of Canada through the Department of Innovation, Science and Economic Development and by the Province of Ontario through the Ministry of Colleges, Universities, Research Excellence and Security.

OpenAI Codex (GPT-6-Astra) and Anthropic Claude Code (Claude Opus 5.5, Claude Fable 5.2 and Claude Sonnet 5.5) were used during the preparation of this manuscript to assist with drafting and English-language editing, manuscript organization, reference checking, the code for the symbolic and numerical calculations, and the plotting code. The authors verified all derivations, calculations and conclusions and take full responsibility for the content of this article.
\end{acknowledgments}

\section*{Data availability}
No observational or experimental data were used. The code that generates the numerical values and figures of Secs.~\ref{Sec:Averaged} and \ref{Sec:Examples} and Appendixes~\ref{App:Toroidal}--\ref{App:Casimir}, with the documentation of the numerical methods and checks, is provided as Supplemental Material \cite{SupplementalMaterial}.

\appendix

\section{Light-ray congruences and the pointwise theorem}
\label{App:Rigidity}

This appendix proves the statements of Secs.~\ref{Sec:NullGeometry} and \ref{Sec:Splitting}.

\begin{lemma}[Fixed-covector congruences]
\label{lem:fixedp}
Let $p\neq0$. The geodesics \eqref{eq:nullgeodesics} form a smooth, affinely parametrized, twist-free null congruence on $\mathcal T\times\mathbb R^n$ whose generators are spatial translates of one another. Its optical tensor $B_p$ is smooth and bounded on every compact subinterval of $\mathcal T$. Every inextendible generator traverses all of $\mathcal T$; its affine parameter is given by \eqref{eq:affinetraversal}, and it is complete in both directions if and only if both endpoint integrals of $1/\omega_p$ diverge.
\end{lemma}

\begin{proof}
With $t$ as parameter, the generator through $x_0$ is $x^i(t)=x_0^i+\int_{t_0}^t h^{ij}p_j\,\omega_p^{-1}\,ds$, which is defined on all of $\mathcal T$ because the integrand is smooth there. Thus $\partial x/\partial x_0=\mathsf 1$, the Jacobi fields $(0,v)$ with $v\in\ker p$ have the positive-definite Gram matrix $h(t)|_{\ker p}$, and no screen collapse occurs at finite $t$. Self-adjointness of $B_p$ makes the Jacobi Wronskian vanish, so these fields form a Lagrangian Jacobi family. On compact subintervals smoothness bounds $h$, $h^{-1}$, $\omega_p$, $K$ and hence $B_p$. Equation~\eqref{eq:affinetraversal} follows from $dt/d\lambda=\omega_p>0$.
\end{proof}

In particular a finite affine endpoint can occur only as $t$ approaches an endpoint of $\mathcal T$, and two-sided completeness of one generator implies it for all generators with the same covector.

\begin{definition}[Optical saturation]
\label{def:saturation}
A ray $[p]_+$ is optically saturated on $\mathcal T$ if $B_p(t)=0$ for every $t\in\mathcal T$, so that thin bundles of its light rays never expand or shear; equivalently, $\dot h(t)$ vanishes on $\ker p\times\ker p$ for every $t\in\mathcal T$. A line $[p]$ is saturated if one, and hence both, of its rays are saturated. An instant $t_*$ with $B_p(t_*)=0$ is a saturated instant of $[p]$.
\end{definition}

In a parallelly propagated screen frame $\{E_A\}$ the optical tensor of a twist-free congruence obeys the Sachs equation $dB_{AB}/d\lambda=-B_{AC}B_{CB}-\mathcal R_{AB}$, with $\mathcal R_{AB}:=R_{\mu\nu\rho\sigma}E_A^\mu k^\nu E_B^\rho k^\sigma$ \cite{Sachs1961}, whose trace is \eqref{eq:Raychaudhuri}. Along a saturated ray it gives $\mathcal R_{AB}=0$, so the screen components of the tidal tensor vanish as well.

\begin{lemma}[Single frozen hyperplane]
\label{lem:singlehyperplane}
Let $V$ be the translation space, $p\in V^*$ nonzero and $S$ a symmetric bilinear form on $V$. Then $S|_{\ker p\times\ker p}=0$ if and only if $S=p\odot\eta$ for a unique $\eta\in V^*$.
\end{lemma}

\begin{proof}
Choose $u$ with $p(u)=1$ and set $\eta(v):=S(u,v)-\frac12S(u,u)p(v)$. Writing $v=v_\perp+p(v)u$ with $v_\perp\in\ker p$, and similarly for $w$, bilinearity gives $S(v,w)=(p\odot\eta)(v,w)$. The converse is immediate. If $p\odot\zeta=0$, evaluation on $(u,v_\perp)$ and on $(u,u)$ gives $\zeta=0$.
\end{proof}

For the normal form \eqref{eq:normalform}, $K(\partial_A,\partial_B)\equiv0$, and Eq.~\eqref{eq:RicciGeneral} gives
\begin{equation}
 R_{00}=-\frac{\ddot a}{a}-2|v|^2,
 \qquad
 R_{AB}=-2v_Av_B .
 \label{eq:normalRicci}
\end{equation}

\paragraph*{Linear form of the Raychaudhuri equation.}
The Raychaudhuri equation has a convenient linear form. Set
\begin{equation}
 x_p:=(\mathcal V\omega_p)^{1/m}>0 .
 \label{eq:potential}
\end{equation}
By \eqref{eq:Thetadef}, $x_p'/x_p=\Theta_p/m$, where a prime denotes $d/d\lambda$, and \eqref{eq:Raychaudhuri} becomes
\begin{equation}
 x_p''=-\frac{x_p}{m}\Big(\operatorname{tr}(\sigma_p^2)+R_{\mu\nu}k^\mu k^\nu\Big).
 \label{eq:potentialRay}
\end{equation}
Thus $x_p$ is concave in the affine parameter wherever $R_{\mu\nu}k^\mu k^\nu\ge0$.

\begin{proof}[Proof of Theorem~\ref{thm:main}]
By Lemma~\ref{lem:fixedp}, the fixed-$p$ congruence is regular, $\gamma$ traverses all of $\mathcal T$, and every generator is a translate of $\gamma$, so (i) holds on all of them. The affine parameter of $\gamma$ ranges over $\mathbb R$. By \eqref{eq:potentialRay} and (i), $x_p$ is a positive concave function on $\mathbb R$, hence constant. Then $\operatorname{tr}(\sigma_p^2)+R_{\mu\nu}k^\mu k^\nu\equiv0$, both terms vanish, and $\Theta_p=m\,x_p'/x_p=0$. Thus $B_p\equiv0$, and by \eqref{eq:hyperplanefreezing} and Lemma~\ref{lem:singlehyperplane} the metric has the normal form \eqref{eq:normalform}.

In that form we may take $E_A=\partial_A$, and \eqref{eq:normalRicci} gives Eq.~\eqref{eq:concavity} for every $h$-unit $e=e^A\partial_A\in\ker p$ and $\ell=\partial_t+e$, together with
\begin{equation}
 (n-1)\mathcal N_p
    =-(n-1)\frac{\ddot a}{a}-2n|v|^2 .
 \label{eq:concavityN}
\end{equation}
Condition (ii) therefore gives $\ddot a\le-2|v|^2a\le0$, so $a$ is positive and concave on $\mathcal T$. If $\mathcal T$ had a finite endpoint, the tangent-line bound would keep $a$ bounded near it, and $d\lambda=a\,dt$ would give a finite affine endpoint, contradicting completeness. Hence $\mathcal T=\mathbb R$, $a$ is a positive concave function on $\mathbb R$ and therefore constant, and (ii) with \eqref{eq:concavity} forces $v=0$, that is $\dot b=0$. The metric \eqref{eq:normalform} is constant.
\end{proof}

\begin{proof}[Proof of Corollary~\ref{cor:ncc}]
The NCC implies (i) and (ii). Null contraction removes the trace and cosmological terms from Einstein's equation, and $T_{0i}=0$ because $R_{0i}=0$.
\end{proof}

\begin{remark}[Relation to splitting theorems]
\label{rem:timelike}
For the comoving congruence $u=\partial_t$ the screen is the whole spatial slice. If $\mathcal T=\mathbb R$ and $R_{00}\ge0$, the timelike Raychaudhuri equation makes $\mathcal V^{1/n}$ positive and concave on $\mathbb R$, hence constant, and the shear and expansion vanish; this is the Bianchi I form of Lorentzian splitting \cite{Eschenburg1988,Galloway1989,Newman1990}. A null screen omits the propagation direction, which is why Theorem~\ref{thm:main} needs its transverse condition. On $\Sigma=\mathbb R^n$ every complete fixed-covector generator is a null line (Lemma~\ref{lem:nulllines}), and the first step of the proof makes the level sets of the optical function $u_p$ totally geodesic, as in Galloway's null splitting theorem \cite{Galloway2000}. That step needs the curvature condition only along the chosen generator, and homogeneity supplies the regular congruence without an achronality assumption or a maximum principle. The argument therefore needs neither the full NCC nor null completeness of the whole spacetime. Under the full NCC and null completeness on $\mathbb R^n$, Galloway's theorem gives a second route to staticity: every generator lies in a totally geodesic null hypersurface, whose screen tidal tensor vanishes by the Sachs equation, so $B_p'=-B_p^2$ on the whole affine line and hence $B_p=0$, and three saturated lines force staticity by Proposition~\ref{prop:classification}.
\end{remark}

Theorem~\ref{thm:main} does not say which end of a geodesic is incomplete. The same congruences give a one-sided statement.

\begin{proposition}[One-sided focusing]
\label{prop:onesided}
Let $\gamma$ be a fixed-$p$ generator and $t_0\in\mathcal T$. If $\Theta_p(t_0)>0$ and $R_{\mu\nu}k^\mu k^\nu\ge0$ along $\gamma$ for $t\le t_0$, then $\gamma$ is past incomplete, and its affine length to the past of $t_0$ is at most $(n-1)/\Theta_p(t_0)$. If $\Theta_p(t_0)<0$ and $R_{\mu\nu}k^\mu k^\nu\ge0$ for $t\ge t_0$, then $\gamma$ is future incomplete, with future affine length at most $(n-1)/|\Theta_p(t_0)|$.
\end{proposition}

\begin{proof}
For $t\le t_0$, Eq.~\eqref{eq:Raychaudhuri} gives $\Theta_p'\le-\Theta_p^2/m$, so $\Theta_p$ stays positive and grows toward the past, and $(1/\Theta_p)'\ge1/m$. Hence $1/\Theta_p$ would reach zero within affine distance $m/\Theta_p(t_0)$ to the past. Since $\Theta_p$ is finite at every $t\in\mathcal T$ by Lemma~\ref{lem:fixedp}, the generator reaches an endpoint of $\mathcal T$ no later than that affine bound. The future case is the time reverse.
\end{proof}

\begin{corollary}[Principal-rate criterion]
\label{cor:eigen}
Suppose the NCC holds. If $\kappa_1+\cdots+\kappa_{n-1}>0$ at some $t_0$, every null geodesic is past incomplete. If $\kappa_2+\cdots+\kappa_n<0$ at some $t_0$, every null geodesic is future incomplete.
\end{corollary}

\begin{proof}
By \eqref{eq:Thetadef}, $\Theta_p=\omega_p\bigl(\operatorname{tr}\mathsf K-\langle q,\mathsf Kq\rangle_h\bigr)$, and $\kappa_1\le\langle q,\mathsf Kq\rangle_h\le\kappa_n$ for every unit $q$. The first hypothesis gives $\Theta_p(t_0)>0$ for every $p$, and the second gives $\Theta_p(t_0)<0$. Proposition~\ref{prop:onesided} applies to every null geodesic.
\end{proof}

In $3+1$ dimensions the first hypothesis of Corollary~\ref{cor:eigen} is strict 2-convexity, under which closed Cauchy surfaces other than spherical space forms are known to give past null incompleteness \cite{GallowayLing2018,LingRossdeutscherSimonSteinbauer2025}; on tori, the existence of a past-incomplete null geodesic is a special case of that theorem. The criterion uses the screen rate $\operatorname{tr}\mathsf K-K_{ij}q^iq^j$ and the NCC, whereas the argument of Borde, Guth and Vilenkin uses the longitudinal rate $K_{ij}q^iq^j=d(1/\omega_p)/d\lambda$ and no energy condition \cite{BordeGuthVilenkin2003,KinneyBGVNCC2026}.

\section{Averaged null curvature}
\label{App:AveragedProofs}

This appendix proves the statements of Sec.~\ref{Sec:Averaged}. For general integrands the limiting prescription matters. For $f=\tanh\lambda$ the symmetric limit vanishes, the offset path $V=U+c$ gives $c$, and the independent lower limit is $-\infty$.

We call a $C^1$ map $s\mapsto(U(s),V(s))\in(0,\infty)^2$, $s\ge0$, an admissible cutoff path if $U'\ge0$, $V'\ge0$, $U,V\to\infty$, and, with the weight $w:=U'V'/(U'+V')$, set to zero where $U'+V'=0$,
\begin{equation}
 \int_0^\infty w(s)\,ds=\infty .
 \label{eq:admissible}
\end{equation}
Straight lines $(U,V)=(u_0+c_1 s,v_0+c_2 s)$ with $u_0,v_0,c_1,c_2>0$ and power-law paths such as $(U,V)=(1+s,(1+s)^\vartheta)$ with $\vartheta>0$ are admissible. Paths that move one cutoff at a time have $w\equiv0$ and are excluded, and the exclusion is necessary (Appendix~\ref{App:Staircase}).

\begin{proof}[Proof of Proposition~\ref{thm:path}]
Write $J(s)=J(U(s),V(s))$ and $D(s)=\Theta(-U(s))-\Theta(V(s))$, so that $I=D-J$ by \eqref{eq:affinesumrule}. By \eqref{eq:Raychaudhuri} and the elementary inequality $U'x^2+V'y^2\ge w\,(x-y)^2$,
\begin{equation}
 J'(s)=U'\operatorname{tr}B^2(-U)+V'\operatorname{tr}B^2(V)\ge\frac{w}{m}\,D^2 .
 \label{eq:pathRiccati}
\end{equation}
If $J^{\mathrm{tot}}<\infty$, integration gives $\int_0^\infty wD^2\,ds\le mJ^{\mathrm{tot}}$. Since $\int w=\infty$, $D^2$ cannot stay above any positive constant on a tail, so $D(s_j)\to0$ on a sequence $s_j\to\infty$, and $I(s_j)\to-J^{\mathrm{tot}}$. If $J^{\mathrm{tot}}=\infty$, suppose $|D|\ge\epsilon J$ on a tail. Then $J'\ge(\epsilon^2w/m)J^2$, so $(1/J)'\le-\epsilon^2w/m$, and $1/J$ would become negative at finite $s$, which is impossible. Hence $|D(s_j)|/J(s_j)\to0$ on a sequence, $I/J\to-1$ there, and $I(s_j)\to-\infty$. This proves \eqref{eq:pathliminf}. The diagonal path $U=V=1+s$ is admissible, and the infimum over $U,V\ge R$ cannot exceed values along it, which gives $\underline I\le-J^{\mathrm{tot}}$. An existing limit must agree with the subsequential limits found above. Finally, a nonnegative lower limit forces $J^{\mathrm{tot}}=0$, hence $B\equiv0$ by continuity and self-adjointness, and $R_{\mu\nu}k^\mu k^\nu\equiv0$ by \eqref{eq:Raychaudhuri}.
\end{proof}

\begin{proposition}[Smeared sum rule]
\label{thm:smeared}
Under the hypotheses of Proposition~\ref{thm:path}, let $\chi\in C_c^\infty(\mathbb R)$ be real with $\chi(0)=1$, let $\chi_\tau$, $C_\chi$ and $I_\chi(\tau)$ be as in Sec.~\ref{Sec:Averaged}, and let $J_\chi(\tau):=\int_{\mathbb R}\chi_\tau^2\operatorname{tr}(B^2)\,d\lambda$. Then the identity \eqref{eq:smearedidentity} holds for every $\tau>0$, and consequently so do the limit \eqref{eq:smearedlimit} and the bound
\begin{equation}
 J_\chi(\tau)\le\left(
 \sqrt{\frac{mC_\chi}{\tau}}
 +\sqrt{\frac{mC_\chi}{\tau}-I_\chi(\tau)}\right)^2 .
 \label{eq:smearedfinite}
\end{equation}
\end{proposition}

\begin{proof}
Compact support allows integration by parts in \eqref{eq:Raychaudhuri} without endpoint terms,
\begin{equation}
 I_\chi(\tau)=-J_\chi(\tau)+2\int_{\mathbb R}\chi_\tau\chi_\tau'\,\Theta\,d\lambda,
 \label{eq:smearedcross}
\end{equation}
and completing the square with $\operatorname{tr}\mathsf 1=m$ and $\int\chi_\tau'^2\,d\lambda=C_\chi/\tau$ gives \eqref{eq:smearedidentity}. By Cauchy--Schwarz and $\Theta^2\le m\operatorname{tr}(B^2)$, the cross term in \eqref{eq:smearedcross} is bounded by $2(mC_\chi J_\chi(\tau)/\tau)^{1/2}$. Since $\chi_\tau^2\to1$ pointwise, dominated convergence gives $J_\chi(\tau)\to J^{\mathrm{tot}}$ when $J^{\mathrm{tot}}<\infty$, and Fatou's lemma gives $J_\chi(\tau)\to\infty$ otherwise; in both cases \eqref{eq:smearedlimit} follows. For \eqref{eq:smearedfinite}, the Frobenius-norm triangle inequality in $L^2(\mathbb R)$ gives $\|\chi_\tau B\|\le\|\chi_\tau B-\chi_\tau'\mathsf 1\|+\|\chi_\tau'\mathsf 1\|$, and \eqref{eq:smearedidentity} identifies the squares of the two terms on the right as $mC_\chi/\tau-I_\chi(\tau)$ and $mC_\chi/\tau$.
\end{proof}

\begin{remark}[Regularity of the congruence]
\label{rem:conevertex}
Propositions~\ref{thm:path} and \ref{thm:smeared} apply to a congruence that is regular along all of $\gamma$. In Minkowski space the light cone of a point has $B=\lambda^{-1}\mathsf 1$ along each radial generator, which is complete and has $R_{\mu\nu}k^\mu k^\nu=0$; the cone congruence fails the hypothesis at its vertex, while the plane-wave congruence through the same generator satisfies it and has $B=0$.
\end{remark}

\begin{corollary}[Averaged saturation in Bianchi I]
\label{cor:bianchiaveraged}
Let a fixed-$p$ generator of a Bianchi I spacetime be complete in both directions. If $\underline I_p\ge0$, or the lower limit of $I_p$ along one admissible cutoff path is nonnegative, or $\lim_{\tau\to\infty}I_{p,\chi}(\tau)\ge0$ for one weight $\chi$, then $B_p\equiv0$ and $R_{\mu\nu}k^\mu k^\nu\equiv0$ on $\mathcal T$, and $\dot h$ vanishes on $\ker p$ for all $t\in\mathcal T$. Theorem~\ref{thm:main} therefore remains true when condition~(i) is replaced by any of these averaged conditions.
\end{corollary}

\begin{proof}
Lemma~\ref{lem:fixedp} provides the regular congruence and shows that the generator traverses all of $\mathcal T$; homogeneity extends the conclusion to all generators with covector $p$. Propositions~\ref{thm:path} and \ref{thm:smeared} give $B_p\equiv0$, and \eqref{eq:hyperplanefreezing} gives the frozen screen. The proof of Theorem~\ref{thm:main} used (i) only to obtain $B_p\equiv0$.
\end{proof}

When cosmic time is complete, condition (ii) of Theorem~\ref{thm:main} can be replaced by an average of $\mathcal N_p$ along the comoving worldlines. Write $c_n=2n/(n-1)$.

\begin{corollary}[Fully averaged rigidity]
\label{cor:fullyaveraged}
Let $\mathcal T=\mathbb R$, and let a fixed-$p$ generator be complete in both directions and satisfy one of the averaged conditions of Corollary~\ref{cor:bianchiaveraged}. Suppose that
\begin{equation}
 \liminf_{s\to\infty}\int_{-U(s)}^{V(s)}\mathcal N_p(t)\,dt\ge0
 \label{eq:transverseaveraged}
\end{equation}
along one admissible cutoff path in cosmic time, or that $\lim_{\tau\to\infty}\int_{\mathbb R}\chi_\tau(t)^2\,\mathcal N_p(t)\,dt\ge0$ for one weight $\chi$ as in Proposition~\ref{thm:smeared}. Then $h$ is constant.
\end{corollary}

\begin{proof}
By Corollary~\ref{cor:bianchiaveraged} the metric has the normal form \eqref{eq:normalform} on $\mathbb R$. With $H=\dot a/a$, Eq.~\eqref{eq:concavityN} reads $\mathcal N_p=-\dot H-H^2-c_n|v|^2$, so
\begin{equation}
 \int_{-U}^{V}\mathcal N_p\,dt=H(-U)-H(V)-\int_{-U}^{V}\bigl(H^2+c_n|v|^2\bigr)dt .
 \label{eq:transversesumrule}
\end{equation}
This has the form of \eqref{eq:affinesumrule} with $m=1$, $\Theta$ replaced by $H$ and $\operatorname{tr}(B^2)$ by $H^2+c_n|v|^2\ge H^2$. The proofs of Propositions~\ref{thm:path} and \ref{thm:smeared} use only this structure, so the lower limit along any admissible path and the smeared limit are at most $-\int_{\mathbb R}(H^2+c_n|v|^2)\,dt$. The hypothesis forces $H\equiv0$ and $v\equiv0$, so $a$ and $b$ are constant.
\end{proof}

Completeness of cosmic time is assumed because, without the concavity of $a$, completeness of the generator alone no longer excludes a finite endpoint of $\mathcal T$. In the proof the scale factor $a$ of the normal form \eqref{eq:normalform} plays the role of the light bundle. Since $t$ is proper time along the comoving observers and $\mathcal N_p=8\pi G(\rho+P_\perp)$ in Einstein gravity, \eqref{eq:transverseaveraged} is an averaged energy condition along a timelike worldline. Quantum inequalities along worldlines give lower bounds, with state and geometric restrictions \cite{FewsterRoman2003,Fewster2000}; they do not by themselves imply the nonnegativity assumed here.

\subsection{A non-admissible cutoff path}
\label{App:Staircase}

The exclusion of cutoff paths that move one cutoff at a time is necessary: the following complete flat FLRW metric has a $C^1$ cutoff path with $\int w\,ds<\infty$ along which the lower limit of $I$ is at least $1-J^{\mathrm{tot}}$.
Let $\psi\in C_c^\infty(\mathbb R)$ satisfy $0\le\psi\le1$, $\psi=1$ on $[-\frac14,\frac14]$ and $\operatorname{supp}\psi\subset[-\frac12,\frac12]$, and set
\begin{equation}
 \Theta(\lambda)=\sum_{k\ge1}\Bigl[\psi\bigl(2^k(\lambda+4^k)\bigr)-\psi\bigl(2^k(\lambda-4^k)\bigr)\Bigr].
 \label{eq:staircaseTheta}
\end{equation}
The bumps have disjoint supports, $\Theta\ge0$ for $\lambda<0$ and $\Theta\le0$ for $\lambda>0$, $\Theta=\pm1$ on intervals of half-width $2^{-k-2}$ around $\mp4^k$, and $\int|\Theta|\,d\lambda\le2$. Let $a(\lambda)=\exp\bigl(\frac12\int_0^\lambda\Theta\bigr)$, so that $e^{-1/2}\le a\le1$, and consider the flat FLRW metric in $3+1$ dimensions with this scale factor along a null geodesic with covector $p$, related to cosmic time by $dt/d\lambda=|p|/a$. Since $a$ is bounded above and below, $t(\lambda)$ is a diffeomorphism of $\mathbb R$ and the spacetime is null and timelike complete, as in Sec.~\ref{Sec:SharpExamples}. The expansion of the fixed-covector congruence through this geodesic is $\Theta=2a'/a$, its shear vanishes, and $\operatorname{tr}(B^2)=\Theta^2/2$, so $J^{\mathrm{tot}}\le1$.

Put $U_k=V_k=4^k$. The path holds $U=U_k$ while $V$ increases from $V_k$ to $V_{k+1}$, then holds $V=V_{k+1}$ while $U$ increases from $U_k$ to $U_{k+1}$, and repeats. On the first leg $\Theta(-U)=1$ and $\Theta(V)\le0$; on the second $\Theta(V)=-1$ and $\Theta(-U)\ge0$. Near each corner replace the path by a $C^1$ arc on which $U$ and $V$ change by less than $2^{-k-4}$, so that both stay on their plateaus. Hence $D=\Theta(-U)-\Theta(V)\ge1$ along the whole path. The weight vanishes on the straight legs, and on each arc $\int w\,ds\le\min(\Delta U,\Delta V)\le2^{-k-4}$ because $w\le\min(U',V')$, so $\int w\,ds<\infty$. By \eqref{eq:affinesumrule}, $I=D-J\ge1-J^{\mathrm{tot}}$ along the path, whereas Proposition~\ref{thm:path} gives a lower limit at most $-J^{\mathrm{tot}}$ along every admissible path.

\section{Saturated lines}
\label{App:SaturationProofs}

This appendix proves the statements of Sec.~\ref{Sec:Saturation}.

\begin{proof}[Proof of Lemma~\ref{lem:energy}]
At $t_*$ the screen compression of $\mathsf K$ vanishes on $\ker p$ by \eqref{eq:hyperplanefreezing}, so in the frame $\{E_A,q\}$ it has the form \eqref{eq:normalK}. Then $(\operatorname{tr}\mathsf K)^2-\operatorname{tr}(\mathsf K^2)=H^2-(H^2+2|v|^2)$, and \eqref{eq:Hamiltonian} gives the result.
\end{proof}

\begin{proposition}[Two saturated lines]
\label{prop:twolines}
Let $p,r\in V^*$ be linearly independent. Both lines $[p]$ and $[r]$ are saturated on $\mathcal T$ if and only if $\dot h(t)=f(t)\,p\odot r$ for a scalar function $f$. Equivalently, with $\Phi(t_0)=0$ and $\dot\Phi=f$, the metric has the form \eqref{eq:twolinefamily}, and positive definiteness restricts $\Phi$ to an interval.
\end{proposition}

\begin{proof}
In a fixed coframe with $e^1=p$ and $e^2=r$, saturation of $[p]$ kills every component $\dot h_{ij}$ with $i,j\ne1$ and saturation of $[r]$ every component with $i,j\ne2$, leaving only $\dot h_{12}=\dot h_{21}$. The converse is immediate.
\end{proof}

Rescaling $p\mapsto c_1p$, $r\mapsto c_2r$ rescales $\Phi$ by $1/(c_1c_2)$ and leaves the metric unchanged. With the $h_0^{-1}$ inner products $\langle p,p\rangle$, $\langle r,r\rangle$ and $\langle p,r\rangle$, where $\langle p,r\rangle^2<\langle p,p\rangle\langle r,r\rangle$, the operator $h_0^{-1}(p\odot r)$ has the two nonzero eigenvalues
\begin{equation}
 \mu_\pm=\langle p,r\rangle\pm\sqrt{\langle p,p\rangle\langle r,r\rangle},
 \qquad
 \mu_+>0>\mu_- .
 \label{eq:kappapm}
\end{equation}
In a fixed $h_0$-orthonormal eigen-coframe $\{\vartheta^+,\vartheta^-,\allowbreak\vartheta^3,\ldots,\vartheta^n\}$,
\begin{equation}
 \begin{aligned}
 h(t)&=[1+\mu_+\Phi(t)](\vartheta^+)^2\\
     &\quad+[1+\mu_-\Phi(t)](\vartheta^-)^2\\
     &\quad+\sum_{A=3}^{n}(\vartheta^A)^2,\\
 -\frac{1}{\mu_+}&<\Phi<-\frac{1}{\mu_-},
 \end{aligned}
 \label{eq:spectralnormalform}
\end{equation}
so every two-line family is diagonalized by one time-independent transformation, and the saturated lines are represented by $\pi_\pm=\sqrt{\mu_+}\,\vartheta^+\pm\sqrt{-\mu_-}\,\vartheta^-$, with $p\odot r=\frac12\pi_+\odot\pi_-$. A single saturated line, by contrast, allows evolution that is non-diagonal in every fixed basis (Appendix~\ref{App:Pulses}).

\begin{proposition}[No third line]
\label{prop:nothirdline}
If $S=p\odot r$ with $p,r$ independent, the only lines $[s]$ with $S|_{\ker s\times\ker s}=0$ are $[p]$ and $[r]$.
\end{proposition}

\begin{proof}
Lemma~\ref{lem:singlehyperplane} writes $S=s\odot\eta$. Since $S$ has rank two, its image is both $\operatorname{span}\{p,r\}$ and $\operatorname{span}\{s,\eta\}$, so $s=c_1p+c_2r$. Choose $x,y\in V$ with $p(x)=r(y)=1$ and $p(y)=r(x)=0$. Then $v=c_2x-c_1y\in\ker s$ and $S(v,v)=-2c_1c_2$, so $c_1c_2=0$.
\end{proof}

\begin{proof}[Proof of Proposition~\ref{prop:classification}]
If $[p]$ and $[r]$ are saturated, Proposition~\ref{prop:twolines} gives $\dot h=f\,p\odot r$ with $f(t_*)\ne0$ at some $t_*$. Any further saturated line must satisfy the compression condition at $t_*$, which Proposition~\ref{prop:nothirdline} excludes. Antipodal symmetry follows from $\ker(-p)=\ker p$.
\end{proof}

\subsection{Pointwise classification}
\label{App:ScreenCompression}

Proposition~\ref{prop:classification} has the following pointwise counterpart. Let $\mathsf A$ be a self-adjoint operator on a real $n$-dimensional inner-product space, $n\geq2$, and define
\begin{equation}
 \begin{aligned}
  \mathcal E(\mathsf A)&=\{q\in S^{n-1}:P_q\mathsf AP_q=0\},\\
  P_q&=\mathsf 1-q\otimes q .
 \end{aligned}
    \label{eq:abstractexceptionalset}
\end{equation}
If $\mathsf A=0$, then $\mathcal E(\mathsf A)=S^{n-1}$. Suppose $\mathsf A\neq0$ and $q\in\mathcal E(\mathsf A)$, and set $w=P_q\mathsf Aq\in q^\perp$ and $\alpha=\langle q,\mathsf Aq\rangle$. Vanishing of the compression gives
\begin{equation}
    \mathsf A=\alpha\,q\otimes q+q\otimes w+w\otimes q,
    \label{eq:ranktwodecomposition}
\end{equation}
so $\operatorname{rank}\mathsf A\le2$ and $\operatorname{im}\mathsf A\subseteq\operatorname{span}\{q,w\}$. For rank one, $w=0$ and the only solutions are the two orientations along $\operatorname{im}\mathsf A$. For rank two, $w\neq0$, every solution lies in the plane spanned by the two eigenvectors with nonzero eigenvalues, and the block of $\mathsf A$ in the basis $\{q,w/\|w\|\}$ has determinant $-\|w\|^2<0$, so $\mathsf A$ is indefinite. Conversely, let the nonzero eigenvalues be $\mu_+>0>\mu_-$ with unit eigenvectors $e_\pm$. Writing $q=\cos\phi\,e_++\sin\phi\,e_-$, the compression vanishes exactly when $\mu_+\sin^2\phi+\mu_-\cos^2\phi=0$, which has the four solutions
\begin{equation}
 \begin{aligned}
 q_{\epsilon_+,\epsilon_-}
   &=\epsilon_+\sqrt{\frac{\mu_+}{\mu_+-\mu_-}}\,e_+\\
   &\quad+\epsilon_-\sqrt{\frac{-\mu_-}{\mu_+-\mu_-}}\,e_-,\\
 \epsilon_\pm&=\pm1 .
 \end{aligned}
    \label{eq:fourexceptionaldirections}
\end{equation}
Rank-one operators therefore give one antipodal pair, rank-two semidefinite operators give none, rank-two indefinite operators give exactly four oriented solutions, and operators of rank at least three give none.

\subsection{Diagonal metrics}
\label{Sec:Diagonal}

For $h=\operatorname{diag}(a_1^2,\ldots,a_n^2)$ with $H_i=\dot a_i/a_i$, the expansion operator is $\mathsf K=\operatorname{diag}(H_1,\ldots,H_n)$ in the orthonormal frame, and Eq.~\eqref{eq:diagonaldirection} gives the frequency and the unit direction. The support of $p$ is time independent. For $r\notin\operatorname{supp}p$ the vector $e_{\hat r}$ lies in the screen, so saturation requires $H_r=0$. For two support indices, the screen vector $\nu_je_{\hat\imath}-\nu_ie_{\hat\jmath}$ gives $H_i\nu_j^2+H_j\nu_i^2=0$. With three or more support indices, the variables $x_i=H_i/\nu_i^2$ satisfy $x_i+x_j=0$ for every pair, so all vanish. The saturated directions of a non-static diagonal metric are therefore of two kinds:
\begin{align}
\operatorname{supp}p=\{i\}:&\quad H_r\equiv0\ \ (r\neq i),
    \label{eq:onesupportclassification}\\
\operatorname{supp}p=\{i,j\}:&\quad H_r\equiv0\ \ (r\notin\{i,j\}),
 \notag\\
 &\quad
    \frac{a_i^2(t)}{p_i^2}+\frac{a_j^2(t)}{p_j^2}=C_{ij}>0 ,
    \label{eq:ellipserelation}
\end{align}
the last condition being the integrated form of $H_ia_i^2/p_i^2+H_ja_j^2/p_j^2=0$. A saturated direction with three or more nonzero momentum components forces a diagonal metric to be static.
In $3+1$ dimensions a saturated direction of a non-static diagonal metric therefore requires at least one scale factor to be constant.

\subsection{Proofs of Corollaries~\ref{cor:strict} and \ref{cor:averagedrigidity}}

\begin{proof}[Proof of Corollary~\ref{cor:strict}]
If $J_p^{\mathrm{tot}}=0$ the line $[p]$ would be saturated. Propositions~\ref{thm:path} and \ref{thm:smeared} with Lemma~\ref{lem:fixedp} give the rest.
\end{proof}

\begin{proof}[Proof of Corollary~\ref{cor:averagedrigidity}]
Corollary~\ref{cor:bianchiaveraged} saturates the three lines, and Proposition~\ref{prop:classification} forces $\dot h\equiv0$. Then $\omega_p$ is constant for every $p$, and a finite endpoint of $\mathcal T$ would be reached at finite affine parameter by the complete generators. For the last statement, null completeness makes every generator two-sided complete, and Lemma~\ref{lem:nulllines} identifies it as a null line, so the achronal condition applies to the generators of any three distinct covector lines.
\end{proof}

\subsection{Proof of Proposition~\ref{prop:tax}}

In the $h$-orthonormal eigenframe of $\mathsf K$ let $x_i=q_i^2$, so that $x_i\ge0$ and $\sum_ix_i=1$, and let $m=\sum_i\kappa_ix_i=\langle q,\mathsf Kq\rangle_h$. With $\mathsf P=\mathsf 1-q\otimes q^\flat$, $\mathsf P\mathsf K\mathsf P=\mathsf K-q\otimes(\mathsf Kq)^\flat-(\mathsf Kq)\otimes q^\flat+m\,q\otimes q^\flat$, and its squared Frobenius norm is
\begin{equation}
 F=\operatorname{tr}(\mathsf K^2)-2|\mathsf Kq|_h^2+m^2
  =\sum_i\kappa_i^2-2\sum_i\kappa_i^2x_i+m^2 ,
 \label{eq:Fexpanded}
\end{equation}
which is \eqref{eq:diagonalbulk} at one instant. Using $\sum_ix_i=1$, this can be rewritten as
\begin{equation}
 \begin{aligned}
 F&=\kappa_{\mathrm{mid}}^2+\bigl(m-\kappa_{\max}-\kappa_{\min}\bigr)^2\\
  &\quad+2(\kappa_{\max}-\kappa_{\mathrm{mid}})(\kappa_{\mathrm{mid}}-\kappa_{\min})\,x_{\mathrm{mid}} .
 \end{aligned}
 \label{eq:Fdecomposition}
\end{equation}
The last term is nonnegative and vanishes for $x_{\mathrm{mid}}=0$, and $m$ ranges over $[\kappa_{\min},\kappa_{\max}]$, with every value attained at $x_{\mathrm{mid}}=0$. The distance from $\kappa_{\max}+\kappa_{\min}$ to this interval is zero if $\kappa_{\min}\le0\le\kappa_{\max}$, $\kappa_{\min}$ if $\kappa_{\min}>0$ and $|\kappa_{\max}|$ if $\kappa_{\max}<0$, which gives \eqref{eq:taxfloor} and the minimizing directions stated after Proposition~\ref{prop:tax}. For the maximum, $F(e_i)=\operatorname{tr}(\mathsf K^2)-\kappa_i^2$ along the $i$th principal axis, and
\begin{equation}
 \sum_ix_iF(e_i)-F=\sum_ix_i\kappa_i^2-m^2\ge0
\end{equation}
by Jensen's inequality, so $F\le\max_iF(e_i)$, which is \eqref{eq:taxceiling}. For the average, the moments $\langle q_iq_j\rangle=\delta_{ij}/3$ and $\langle q_iq_jq_kq_l\rangle=(\delta_{ij}\delta_{kl}+\delta_{ik}\delta_{jl}+\delta_{il}\delta_{jk})/15$ of the uniform measure on the unit sphere give $\langle F\rangle=[7\operatorname{tr}(\mathsf K^2)+\theta^2]/15$, which is \eqref{eq:taxaverage} because $\operatorname{tr}(\mathsf K^2)=\sigma^2+\theta^2/3$. Finally \eqref{eq:taxintegrated} follows from \eqref{eq:Jcosmictime} and $F\ge\mu(\mathsf K)$. \qed

\section{Compact sections}
\label{App:Toroidal}

Let $\Gamma$ be a discrete group of translations of $V=\mathbb R^n$ and $\Sigma=\mathbb R^n/\Gamma$. Translations are isometries of every $h(t)$, so the projection $\pi:\mathcal T\times\mathbb R^n\to\mathcal T\times\Sigma$ is a local isometry. Along a future-directed causal curve $t$ increases, so the curve can be parametrized by $t$. Timelike curves lift, so $\pi(P)$ and $\pi(Q)$ with $t(P)<t(Q)$ are chronologically related in the quotient if and only if $P\ll Q+L$ on the cover for some $L\in\Gamma$. For $\xi\in V^*$ let $\omega_\xi=(h^{ij}\xi_i\xi_j)^{1/2}$, as in \eqref{eq:omegadef}.
We write $W=\operatorname{span}\Gamma$ and $W^\circ=\{p:p(W)=0\}$ for its annihilator.

\begin{lemma}[Null lines]
\label{lem:nulllines}
Let $\Sigma=\mathbb R^n/\Gamma$ for a discrete group $\Gamma$ of translations, with $\Gamma=\{0\}$ for $\Sigma=\mathbb R^n$, and let $p(\Gamma)=0$. Every two-sided complete fixed-$p$ generator is a null line of $\mathcal T\times\Sigma$.
\end{lemma}

\begin{proof}
Since $u_p(t,x+L)=u_p(t,x)+p(L)=u_p(t,x)$ for $L\in\Gamma$, the optical function \eqref{eq:opticalfunction} is defined on $\mathcal T\times\Sigma$. For a future-directed causal vector $X=\dot t\,\partial_t+v^i\partial_i$, Cauchy--Schwarz gives $du_p(X)=\omega_p[h(q,v)-\dot t]\le0$, with strict inequality for timelike $X$, while $u_p$ is constant along the generator. No two of its points can be joined by a timelike curve.
\end{proof}

On compact spatial sections complete generators need not be achronal.

\begin{lemma}[Winding generators]
\label{lem:winding}
Let $\Sigma=\mathbb R^n/\Gamma$ for a discrete group $\Gamma$ of translations, let $\gamma$ be a fixed-$p$ generator with affine parameter $\lambda$, and let $L\in\Gamma$ with $p(L)<0$. If
\begin{equation}
 \int_{\lambda_1}^{\lambda_2}\frac{d\lambda}{h(L,L)}>\frac{1}{2|p(L)|}
 \label{eq:windingcondition}
\end{equation}
for some $\lambda_1<\lambda_2$, then $\gamma(\lambda_1)$ and $\gamma(\lambda_2)$ are joined by a future-directed timelike curve in $\mathcal T\times\Sigma$, so $\gamma$ is not achronal. This holds in particular whenever $\gamma$ is complete in a direction along which $h(L,L)$ stays bounded.
\end{lemma}

\begin{proof}
Let $k=d\gamma/d\lambda$, and let $L$ also denote the Killing field $L^i\partial_i$; by \eqref{eq:pconserved}, $g(k,L)=p(L)$. Put $\phi=1/h(L,L)$, which is smooth and positive along $\gamma$, $\mathcal I=\int_{\lambda_1}^{\lambda_2}\phi\,d\lambda$ and $\vartheta=1/(2|p(L)|\,\mathcal I)$, so that $0<\vartheta<1$ by \eqref{eq:windingcondition}. Let $f(\lambda)=\mathcal I^{-1}\int_{\lambda_1}^\lambda\phi\,d\lambda'$, and let $c(\lambda)$ be the point of the cover with the cosmic time of $\gamma(\lambda)$ and spatial position $x_\gamma(\lambda)+f(\lambda)L$. Then $c(\lambda_1)=\gamma(\lambda_1)$, $c(\lambda_2)=\gamma(\lambda_2)+L$, and $c'=k+f'L$ with $f'h(L,L)=1/\mathcal I=2\vartheta|p(L)|$, so
\begin{equation*}
 g(c',c')=2f'p(L)+f'^2h(L,L)=-2(1-\vartheta)\,|p(L)|\,f'<0 .
\end{equation*}
Since $dt/d\lambda=\omega_p>0$ along $c$, the curve is future-directed and timelike, and its projection joins $\pi(\gamma(\lambda_1))$ to $\pi(\gamma(\lambda_2))$. If $\gamma$ is complete toward increasing $\lambda$ and $h(L,L)\le M$ there, the integral in \eqref{eq:windingcondition} exceeds $1/(2|p(L)|)$ once $\lambda_2-\lambda_1>M/(2|p(L)|)$; the past direction is the same.
\end{proof}

In cosmic time the left-hand side of \eqref{eq:windingcondition} is $\int dt/[\omega_p\,h(L,L)]$, and the timelike curve of the proof is the generator itself, displaced steadily by one lattice vector. In particular, if $\Sigma$ is a torus and $h$ is constant for $t\ge t_*$, no future-complete null geodesic is achronal: $\Gamma$ spans $\mathbb R^n$, so some $L\in\Gamma$ has $p(L)<0$, and $h(L,L)$ is bounded toward the future.

\paragraph*{Chronology on the cover.}
For $t_1<t_2$ in $\mathcal T$ and $D\in V$, the point $(t_2,x+D)$ lies in the chronological future of $(t_1,x)$ on $\mathcal T\times\mathbb R^n$ if and only if
\begin{equation}
 \xi(D)<\int_{t_1}^{t_2}\omega_\xi\,dt
 \qquad\text{for every }\xi\in V^*\setminus\{0\}.
 \label{eq:reach}
\end{equation}
Indeed, a future-directed timelike curve from $(t_1,x)$ has the form $(t,x+y(t))$ with $|\dot y|_h<1$, and $\xi(\dot y)\le\omega_\xi|\dot y|_h<\omega_\xi$ by Cauchy--Schwarz, so integration gives \eqref{eq:reach}. Conversely, the set $\mathcal R$ of reachable $D$ is open, and it is convex because $|\cdot|_{h(t)}$ is a norm for each $t$. The velocities $(1-\epsilon)h^{-1}\xi/\omega_\xi$ show that its support function is $\rho(\xi)=\int_{t_1}^{t_2}\omega_\xi\,dt$, so its closure is $\{D:\xi(D)\le\rho(\xi)\ \forall\xi\}$ and $\mathcal R$ is the interior of that set. If $D$ obeys \eqref{eq:reach}, then $\rho(\xi)-\xi(D)$ has a positive minimum on a unit sphere of $V^*$, so a neighborhood of $D$ lies in the closure and $D\in\mathcal R$.

Two consequences make Lemma~\ref{lem:winding} sharp in the examples of Secs.~\ref{Sec:CasimirTurn} and \ref{Sec:Toroidal}.

\emph{Constant $h(L,\cdot)$.} Suppose $\dot h(L,\cdot)=0$ for all $t$, as for $L\in W$ when $b$ is constant in \eqref{eq:normalform}, or on a static end. Then $\gamma(\lambda_1)$ and $\gamma(\lambda_1+s)+L$ are chronologically related on the cover if and only if $p(L)<0$ and $s>h(L,L)/(2|p(L)|)$. Sufficiency is Lemma~\ref{lem:winding}. For necessity, let $\ell=h(L,\cdot)$, which is a constant covector, let $X$ be the comoving displacement of $\gamma$ over $[\lambda_1,\lambda_1+s]$, and take $\xi=p+\ell/s$. With $h^{-1}(p,\ell)=\ell(h^{-1}p)=p(L)$, $h^{-1}(\ell,\ell)=h(L,L)$, $dt=\omega_p\,d\lambda$ and $(A^2+B)^{1/2}\le A+B/(2A)$ for $A>0$,
\begin{equation*}
 \begin{aligned}
 \int\omega_\xi\,dt&=\int\omega_p\Bigl(\omega_p^2+\frac{2p(L)}{s}+\frac{h(L,L)}{s^2}\Bigr)^{1/2}d\lambda\\
 &\le\int\omega_p^2\,d\lambda+p(L)+\frac{h(L,L)}{2s},\\
 \xi(X+L)&=\int\omega_p^2\,d\lambda+2p(L)+\frac{h(L,L)}{s},
 \end{aligned}
\end{equation*}
so $\xi(X+L)-\int\omega_\xi\,dt\ge p(L)+h(L,L)/(2s)$, which is nonnegative unless $p(L)<0$ and $s>h(L,L)/(2|p(L)|)$; then \eqref{eq:reach} fails. On the products of Proposition~\ref{prop:toroidal}(iii), and in particular on \eqref{eq:dStorus}, two points of a winding generator are therefore chronologically related exactly when their affine separation exceeds $\min_L h(L,L)/(2|p(L)|)$ over $L\in\Gamma$ with $p(L)<0$.

\emph{Conformally flat tori.} If $h=a(t)^2h_0$ with $h_0$ constant, then $\int\omega_\xi\,dt=|\xi|_{h_0^{-1}}\int dt/a$, and \eqref{eq:reach} reads $|D|_{h_0}<\int_{t_1}^{t_2}dt/a$: the cover is conformal to a slab of Minkowski space. The generator with $h_0$-unit comoving direction $n$ moves by $n\int dt/a$, so it has a shortcut through $L$ if and only if $h_0(n,L)<0$ and $\int_{t_1}^{t_2}dt/a>|L|_{h_0}^2/(2|h_0(n,L)|)$. With $N_{\mathrm{cross}}:=\int_{\mathcal T}dt/a$, the generator is therefore achronal if and only if $N_{\mathrm{cross}}\le\min\{|L|_{h_0}^2/(2|h_0(n,L)|):L\in\Gamma,\ h_0(n,L)<0\}$. For $\Gamma=\mathbb Z^3$ and $h_0=\delta$ the minimum is $1/(2\max_i|n_i|)$, because $|L|^2\ge\sum_i|L_i|\ge|n\cdot L|/\max_i|n_i|$ for $L\in\mathbb Z^3$, with equality for $L=-\operatorname{sgn}(n_k)e_k$ and $|n_k|=\max_i|n_i|$; this is the criterion used in Sec.~\ref{Sec:CasimirTurn}. A torus that is large compared with the distance light travels is constrained instead: $a=10(1+t^2)$ on the unit cubic torus has $N_{\mathrm{cross}}=\pi/10<1/2$, so every null geodesic is achronal, while $J_p=\pi|p|/10>0$.

\subsection{Toroidal sections}

Let $\Sigma=\mathbb R^3/\Gamma$, where $\Gamma$ is a lattice of comoving translations of rank $c$, so that $\Sigma\cong\mathbb R^{3-c}\times T^c$. Two-sided complete generators with $p\in W^\circ$ are null lines (Lemma~\ref{lem:nulllines}), and the others are not achronal whenever the torus stays bounded along a complete end (Lemma~\ref{lem:winding}). Write $\ell_e=\partial_t+e$ as in Theorem~\ref{thm:main}.

\begin{proposition}[Toroidal sections]
\label{prop:toroidal}
Let $h$ be nonconstant, let $\mathcal A$ be the set of lines $[p]$ whose generators are two-sided complete and achronal in $\mathcal T\times\Sigma$, and suppose that along the generators of every line in $\mathcal A$ one of the averaged conditions of Corollary~\ref{cor:bianchiaveraged} holds. Then:
\begin{enumerate}
\item[(i)] every line in $\mathcal A$ is saturated, $\mathcal A$ has at most two elements, and if $\mathcal A\neq\emptyset$ then $G_{00}\le0$ at every instant;
\item[(ii)] $\mathcal A$ contains every line of $W^\circ$ with two-sided complete generators, so $c\le1$ is impossible whenever three pairwise distinct lines of $W^\circ$ have complete generators, in particular when the spacetime is null complete;
\item[(iii)] if $c=2$ and the generators of $W^\circ=[p]$ are complete, then $h|_{W\times W}$ is constant, $h$ has the normal form \eqref{eq:normalform} with $dz=p$, $G_{00}=-|v|^2$ at every instant, at some instant $R_{\mu\nu}\ell_e^\mu\ell_e^\nu<0$ for every unit $e\in W$, every winding generator that is complete in one direction is not achronal, so that $\mathcal A=\{[p]\}$, and if $b$ is constant the spacetime is the product of $-dt^2+a^2dz^2$ with the static flat torus $W/\Gamma$, with $G_{00}\equiv0$;
\item[(iv)] if $c=3$, no line to which Lemma~\ref{lem:winding} applies lies in $\mathcal A$; if the lemma applies to every line, the hypothesis is empty.
\end{enumerate}
In Einstein gravity the averaged conditions may be imposed on $T_{\mu\nu}k^\mu k^\nu$, and $G_{00}=8\pi G\rho+\Lambda$. In every case Corollary~\ref{cor:strict} still makes the averaged null curvature negative along every complete null geodesic outside at most two lines.
\end{proposition}

\begin{proof}
The spatial inversion $x\mapsto-x$ is an isometry of every $h(t)$ and maps $\Gamma$ to itself, so it descends to an isometry of $\mathcal T\times\Sigma$ that exchanges the generators of $p$ and $-p$. Completeness (Lemma~\ref{lem:fixedp}), achronality and the averaged conditions, which involve $R_{\mu\nu}k^\mu k^\nu$ and $d\lambda=dt/\omega_p$ only, are therefore properties of the line $[p]$.

(i) Let $[p]\in\mathcal A$. The averaged condition concerns $R_{\mu\nu}k^\mu k^\nu$ along a generator, which is the same on the cover, so Corollary~\ref{cor:bianchiaveraged}, applied on the cover, gives $B_p\equiv0$: the line is saturated. Proposition~\ref{prop:classification} allows at most two saturated lines for nonconstant $h$. If $\mathcal A\neq\emptyset$, Lemma~\ref{lem:energy} applies at every instant and gives $G_{00}=-|v|^2\le0$.

(ii) A line of $W^\circ$ with two-sided complete generators lies in $\mathcal A$ by Lemma~\ref{lem:nulllines}. For $c\le1$, $\dim W^\circ\ge2$; three pairwise distinct lines of $W^\circ$ with complete generators would lie in $\mathcal A$, contradicting (i). Null completeness makes every generator complete.

(iii) For $c=2$, $W^\circ=[p]$ and $\ker p=W$. By (ii) and (i), $[p]$ is saturated, so $\dot h$ vanishes on $W\times W$ at all times and $h|_{W\times W}$ is constant. Lemma~\ref{lem:singlehyperplane} and Sec.~\ref{Sec:SaturationDef} give the normal form \eqref{eq:normalform} with $dz=p$; the constant linear changes of coordinates used there preserve $W=\{dz=0\}$ and carry $\Gamma\subset W$ to a lattice in the same subspace. Lemma~\ref{lem:energy} at every instant gives $G_{00}=-|v|^2$. The generators of $[p]$ are complete and have $R_{\mu\nu}k^\mu k^\nu\equiv0$, so condition (i) of Theorem~\ref{thm:main} holds on the cover; since $h$ is not constant, condition (ii) fails at some instant, at which $R_{\mu\nu}\ell_e^\mu\ell_e^\nu<0$ for every unit $e\in\ker p=W$. If $r\notin[p]$, then $r(L)\neq0$ for some $L\in\Gamma$, because $\Gamma$ spans $W=\ker p$, and we take $r(L)<0$; since $L\in W$, $h(L,L)$ is constant, so by Lemma~\ref{lem:winding} no $r$-generator that is complete in one direction is achronal, and $\mathcal A=\{[p]\}$. If $b$ is constant, the coordinates $y^A=x^A+b_Az$ fix $W$ pointwise and bring $h$ to $\sum_A(dy^A)^2+a^2dz^2$, and $v=0$ gives $G_{00}\equiv0$.

(iv) If Lemma~\ref{lem:winding} applies to a line, its generators are not achronal, so the line is not in $\mathcal A$; if it applies to every line, $\mathcal A=\emptyset$.

The last statement is Corollary~\ref{cor:strict} on the cover, which does not involve achronality.
\end{proof}

The same proof excludes every $c\le n-2$ on quotients $\mathbb R^{n-c}\times T^c$ of $\mathbb R^n$ by comoving translations. In case (iii) with constant $b$, null completeness requires $\mathcal T=\mathbb R$ and $\int a\,dt=\infty$ at both ends for the two-dimensional factor $-dt^2+a^2dz^2$; the metric \eqref{eq:onelineexample} with $x$ and $y$ periodic is an example with varying tilt.

\paragraph*{A twisted quotient.}
Let $g(x,y,z)=(-x,-y,z+\ell)$ with $\ell>0$, and let $h=\operatorname{diag}(X,Y,c^2)$ be the periodic metric \eqref{eq:sharpperiodicmodel}, which is even in $t$. Then $g$ is an isometry of every $h(t)$, the group $\langle g\rangle$ acts freely and properly discontinuously, and $\Sigma=\mathbb R^3/\langle g\rangle$ is diffeomorphic to $\mathbb R^2\times S^1$. For an in-plane covector, $p_3=0$, the generator through $(0,x_0)$ has an in-plane comoving displacement $\Delta(t)=\int_0^tq\,dt'$ that is odd in $t$, so the spatial positions of $g(\gamma(s))$ and $\gamma(-s)$ differ by the constant vector $(-2x_0^1,-2x_0^2,\ell)$ for every $s$. The metric is bounded, so the straight comoving curve between these two points is timelike for large $s$. For $p_3\neq0$, Lemma~\ref{lem:winding} on the double cover $\mathbb R^3/\langle g^2\rangle$, with $L=\pm2\ell e_3$, gives timelike shortcuts. This null-complete, non-static spacetime with one compact direction therefore has no null line, and achronal averaged conditions hold on it vacuously; Proposition~\ref{prop:toroidal} needs $\Gamma$ to consist of translations.

\paragraph*{The frozen two-torus.}
On the metric \eqref{eq:dStorus} every winding direction has $I_p^{\mathrm{ind}}=-\infty$, and by the criterion for constant $h(L,\cdot)$ two points of a winding geodesic are chronologically related exactly when their affine separation exceeds $\min_Lh(L,L)/(2|p(L)|)$ over $L\in\Gamma$ with $p(L)<0$, as on a static torus. Besides the growth of the torus size as $\sinh\mu t$, rectangular deformations of the square grow, at the late-time rate $1.18\,\mu$ for periodic bosons and $1.94\,\mu$ for antiperiodic fermions, while rhombic ones decay, so the square torus is a saddle of the Casimir energy at fixed area. For a rectangular torus the Casimir tension differs along the two sides, and the product is replaced by the tilted family $a=a_0\cosh\alpha t$, $b=b_0+(2a_0v_0/\alpha)\sinh\alpha t$ of the normal form \eqref{eq:normalform}, with $v_0$ along the short side fixed by $v_0^2=4\pi G(P_{\mathrm{long}}-P_{\mathrm{short}})$, $\rho_\Lambda=|\rho_C|-(P_{\mathrm{long}}-P_{\mathrm{short}})/2$ and $\alpha^2=8\pi G(|\rho_C|+P_{\mathrm{short}}-2P_{\mathrm{long}})$. This family is null complete, and the Casimir stress supports it, with $\alpha^2>0$, only for aspect ratios up to $1.306$ for periodic bosons and $1.129$ for antiperiodic fermions. Making $z$ periodic with period $P$ turns \eqref{eq:dStorus} into a three-torus model whose $[dz]$ generators are achronal if and only if $P\ge2\pi/\mu$; for $P=2\pi/\mu$ it is global $dS_2\times T^2$, a three-torus example in which $\mathcal A=\{[dz]\}$ is not empty.

\section{Effective loop quantum cosmology}
\label{App:LQC}

\subsection{Effective equations}

We use the effective Hamiltonian of the improved dynamics of Ref.~\cite{AshtekarWilsonEwing2009} for a Bianchi I universe with a massless scalar field $\phi$, in cosmic time. Let $\mathfrak p_i=a_ja_k$ be the triad components and $\mathfrak c_i$ the connection components, with $\{\mathfrak c_i,\mathfrak p_j\}=8\pi G\gamma\,\delta_{ij}$, and let $\mathcal V=\sqrt{\mathfrak p_1\mathfrak p_2\mathfrak p_3}$, $\bar\mu_i=\sqrt\Delta\,\mathfrak p_i/\mathcal V$, $b_i=\bar\mu_i\mathfrak c_i$ and $\ell=\gamma\sqrt\Delta$. The Hamiltonian constraint is
\begin{equation}
 \mathcal C=-\frac{\mathcal V}{8\pi G\ell^2}\,S+\frac{\Pi_\phi^2}{2\mathcal V}=0,
 \qquad
 S=\sum_{i<j}\sin b_i\sin b_j ,
 \label{eq:LQChamiltonian}
\end{equation}
with scalar density $\rho=\Pi_\phi^2/(2\mathcal V^2)$ and critical density $\rho_c=3/(8\pi G\ell^2)$. With $f_i:=\cos b_i\,(\sin b_j+\sin b_k)$ for $\{i,j,k\}=\{1,2,3\}$, Hamilton's equations give
\begin{equation}
 H_i=\frac{\tfrac12\sum_jf_j-f_i}{\ell},
 \qquad
 \ell\,\dot b_i=\frac12\sum_j(b_j-b_i)f_j-S ,
 \label{eq:LQCeom}
\end{equation}
and $\Pi_\phi$ is constant. These equations conserve $\mathcal V(b_i-b_j)$ and $S\mathcal V^2$, so that $b_i=b+d_i/\mathcal V$ with constants $d_i$, $\sum_id_i=0$. For $d_i=0$ they reduce to the isotropic solution \eqref{eq:LQCiso}, with $H_{\max}=1/(2\ell)$.

\subsection{The triaxial family}

Constant rescalings of the axes set $a_i=1$ at the volume bounce, so that $\mathcal V=1$ there and unit covectors have $\omega_B=1$. We take
\begin{equation}
 \begin{aligned}
 d_i&=d\,e_i(\beta),\\
 e(\beta)&=\Bigl(-\cos\beta-\tfrac{\sin\beta}{\sqrt3},\ \tfrac{2\sin\beta}{\sqrt3},\ \cos\beta-\tfrac{\sin\beta}{\sqrt3}\Bigr),
 \end{aligned}
 \label{eq:LQCfamily}
\end{equation}
with $\beta=\pi\eta/6$, so that $\sum_ie_i=0$ and $\sum_ie_i^2=2$. Let $C(\mathcal V)=\sum_{i<j}\cos[(d_i-d_j)/\mathcal V]$ and $R(\mathcal V)=(3+2C)^{1/2}$. The bounce has $S_B=[C(1)+R(1)]/2$ and $\rho_B/\rho_c=S_B/3$, and with $\hat t=t/\ell$ the volume obeys
\begin{equation}
 \begin{aligned}
 \Bigl(\frac{d\mathcal V}{d\hat t}\Bigr)^2&=E\,F,
 \qquad
 E=S_B+\frac{\mathcal V^2}{2}(R-C),\\
 F&=\frac{C+R}{2}-\frac{S_B}{\mathcal V^2},
 \end{aligned}
 \label{eq:LQCvolume}
\end{equation}
where $F(1)=0$ and $E,F>0$ for $\mathcal V>1$ on the branch $0\le d<\pi/(2\cos\beta)$, so both branches reach arbitrarily large volume. The amplitude is $r=d/(S_B+d^2)^{1/2}\in[0,1)$. At both ends $\mathcal V\simeq v_\infty|\hat t|$ with $v_\infty=[3(S_B+d^2)]^{1/2}$, and the solution approaches a Kasner solution with a scalar field, with exponents
\begin{equation}
 k_i^\pm=\frac13\pm\frac{d_i}{v_\infty},
 \qquad
 \sum_ik_i^\pm=1,
 \qquad
 \sum_i(k_i^\pm)^2=\frac{1+2r^2}{3},
 \label{eq:LQCkasner}
\end{equation}
and $k_i^\pm\ge(1-2r)/3>-\tfrac13$. Thus $r^2=\sigma^2/(\sigma^2+16\pi G\rho)$ at the ends, vacuum Kasner solutions are reached only as $r\to1$, and $\int dt/\omega_p$ diverges at both ends for every $p$. Far from the bounce $\ell^2G_{00}=S_B/\mathcal V^2+O(\mathcal V^{-4})>0$. At the bounce $\theta=0$ and the principal rates are the shear eigenvalues: $\eta=0$ gives $(-s,0,s)$ with $\ell s=\sin d\,(1+\cos d)$, for which the floor \eqref{eq:taxbounce} vanishes on the two lines $[1,0,\pm1]$, and $\eta=\pm1$ gives the axisymmetric spectra. The largest squared shear of the effective dynamics is $\sigma^2_{\max}=27/(8\ell^2)$, Eq.~(2.21) of Ref.~\cite{GuptSingh2012a}; the relation between shear and density at the bounce is studied in Ref.~\cite{McNamaraSainiSingh2023}.

\subsection{Numerical evaluation}

For each background, $J_p$ is computed from \eqref{eq:diagonalbulk} and, independently, $I_p$ from the Ricci contraction
\begin{equation}
 \frac{R_{\mu\nu}k^\mu k^\nu}{\omega_p^2}=-\sum_i(\dot H_i+H_i^2)+\sum_i\nu_i^2\bigl(\dot H_i+\theta H_i\bigr) .
 \label{eq:LQCricci}
\end{equation}
The integrals are evaluated in the variable $u=\operatorname{arcosh}\mathcal V$ with Gauss--Legendre panels on each branch, and the Kasner tails beyond the cutoff are added analytically with the rotation of the direction kept. The grid of Table~\ref{tab:lqc} is $r\in\{0,0.05,0.25,0.5,0.8,0.99\}$ and $\eta\in\{0,\pm0.5,\pm1\}$, with 2304 equal-weight Fibonacci directions per background. Over all $69\,120$ rays the largest $|I_p+J_p|/J_p$ is $7\times10^{-14}$; the Hamiltonian constraint is preserved to $2\times10^{-15}$; an independent integration of Hamilton's equations in cosmic time agrees to $10^{-11}$ in $\ln a_i$; and refining the sky to 9216 directions changes $\langle J\rangle$ by a relative amount below $4\times10^{-7}$. Extremal values over the sky were refined by multistart optimization; they are numerical extrema, not certified ones. The one-parameter family $\eta=0$ was also computed separately on 24 amplitudes up to the vacuum endpoint $r=1$, where $H_{\max}=1/(3\ell)$, $J_{\min}=11.275\,\omega_BH_{\max}=5.08\,J_{\mathrm{iso}}$ and $J_{\max}=19.661\,\omega_BH_{\max}=8.87\,J_{\mathrm{iso}}$.

The critical lengths of the SNEC remark in Sec.~\ref{Sec:Averaged} use windows $\chi_\tau(\lambda)=\chi(\lambda/\tau)$ with $\lambda=0$ at the bounce. For the isotropic bounce and $\mathcal B=1/(32\pi)$ the condition \eqref{eq:SNECequivalence} has a single crossing from holding to failing, at $\tau_c=1.527\,\lambda_b$ for the Gaussian $\chi=e^{-s^2/2}$, at $6.107\,\lambda_b$ for the compactly supported bump $\chi=\exp[1-1/(1-s^2)]$, $|s|<1$, and at $0.755\,\lambda_b$ for the Lorentzian weight $\chi=(1+s^2)^{-1/2}$, where $\lambda_b=\ell/(3\omega_B)$ and the affine curvature radius is $\sqrt{3/2}\,\lambda_b$. The identity \eqref{eq:smearedidentity} extends to the two windows without compact support because $B_p$ decays at both ends.

\section{Casimir stress on flat tori}
\label{App:Casimir}

For a free conformal field on a flat space whose compact directions are generated by lattice vectors $L_{(a)}$, with twist $\alpha_a=0$ for a periodic and $\alpha_a=\tfrac12$ for an antiperiodic cycle, let $\zeta(L)=\cos(2\pi\alpha\cdot n)$ for $L=\sum_an_aL_{(a)}\in\Gamma$, and define the lattice sums
\begin{equation}
 S=\sum_{L\in\Gamma\setminus\{0\}}\frac{\zeta(L)}{|L|^4},
 \qquad
 M_{ij}=\sum_{L\in\Gamma\setminus\{0\}}\frac{\zeta(L)\,L_iL_j}{|L|^6} .
 \label{eq:casimirsums}
\end{equation}
The vacuum stress relative to infinite flat space is
\begin{equation}
 \rho=-\frac{g_*S}{2\pi^2},
 \qquad
 P_{ij}=\frac{g_*}{2\pi^2}\bigl(S\,\delta_{ij}-4M_{ij}\bigr),
 \label{eq:casimirstress}
\end{equation}
with $g_*=1,2,-2,-4$ for a real scalar, a Maxwell field, a Weyl fermion and a Dirac fermion. Since $\operatorname{tr}M=S$ the stress is traceless, and since $M$ has no components along non-compact directions, $P=-\rho$ there: the null energy vanishes for light rays that do not wind. In general $\rho\,\delta_{ij}+P_{ij}=-(2g_*/\pi^2)M_{ij}$. Table~\ref{tab:casimir} lists the energy densities for the sections used in Sec.~\ref{Sec:Examples}.

\begin{table}[t]
\caption{Casimir energy density $\rho a^4$ of free massless fields on flat sections with compact directions of circumference $a$. P and AP denote periodic and antiperiodic conditions on every cycle.}
\label{tab:casimir}
\begin{ruledtabular}
\begin{tabular}{lccc}
Section & Real scalar (P) & Weyl (AP) & Maxwell (P)\\
\hline\noalign{\vskip2pt}
$\mathbb R^2\times S^1$ & $-0.1097$ & $-0.1919$ & $-0.2193$\\
$\mathbb R\times T^2$, square & $-0.3053$ & $-0.3053$ & $-0.6106$\\
$T^3$, cubic & $-0.8375$ & $-0.3914$ & $-1.6751$\\
\end{tabular}
\end{ruledtabular}
\end{table}

On a single circle the scalar value is $\rho=-\pi^2/(90a^4)$ and $T^\mu{}_\nu=\rho\operatorname{diag}(-1,-1,-1,3)$, with the compact direction last. On the square torus $T^\mu{}_\nu=|\rho|\operatorname{diag}(1,1,-1,-1)$ in the order $(t,z,x,y)$ with $z$ non-compact, which is the source of \eqref{eq:dStorus}; the scalar and antiperiodic Weyl values are $-0.3053/a^4$, one third of Catalan's constant in units of $a^{-4}$. On the cubic torus $P_i=\rho/3$ and $\rho+P_i=\tfrac43\rho$, the source of \eqref{eq:casimirturn} with $c_3=-\rho a^4$. On a rectangular torus the null energy $\rho+P$ along the long side is negative for every finite aspect ratio, for periodic bosons and for antiperiodic fermions, and tends to zero only as the aspect ratio tends to infinity.

These are the stresses of the static flat geometry. On the backgrounds of Secs.~\ref{Sec:Toroidal} and \ref{Sec:CasimirTurn} they are the leading term of the renormalized stress tensor; curvature corrections and the trace anomaly, for instance $\langle T^\mu{}_\mu\rangle=\mu^4/(1440\pi^2)$ per conformal scalar on \eqref{eq:dStorus}, are of relative order $\mu^2a_T^2$ and $(H_\infty a_b)^2$ respectively, small for tori much larger than the Planck length, and are not included.

\section{Saturating examples and finite episodes}
\label{App:Pulses}

\subsection{The periodic and one-line models}

For the periodic model \eqref{eq:sharpperiodicmodel} with vanishing cosmological constant, the effective Einstein source has
\begin{equation}
 \begin{aligned}
 8\pi G\,\rho&=-\frac{\varepsilon^2\sin^2 t}{4XY}\le0,\\
 0&=Y(\rho+P_1)+X(\rho+P_2) ,
 \end{aligned}
 \label{eq:sharprho}
\end{equation}
where $P_i=T_{\hat\imath\hat\imath}$. The first relation is the cost of saturation identified in Lemma~\ref{lem:energy}, with $|v|^2=\varepsilon^2\sin^2t/(4XY)$, and the second expresses saturation along the rays \eqref{eq:sharpmomenta}. The principal null combinations are
\begin{align}
 8\pi G(\rho+P_1)&=-\frac{\varepsilon\cos t}{2Y}+\frac{\varepsilon^2\sin^2t}{4Y^2}-\frac{\varepsilon^2\sin^2t}{4XY},
 \nonumber\\*
 8\pi G(\rho+P_2)&=\frac{\varepsilon\cos t}{2X}+\frac{\varepsilon^2\sin^2t}{4X^2}-\frac{\varepsilon^2\sin^2t}{4XY},
 \label{eq:sharpnullcombinations}
\end{align}
and the third principal pressure follows from the same Einstein equations. They satisfy the identity in \eqref{eq:sharprho}. Since $\rho<0$ whenever $\sin t\ne0$, and a principal null combination is negative when $\sin t=0$, the weak and dominant energy conditions fail at every time. For $X_0=Y_0$ and the saturated covector $p=(1,1,0)$, a direct computation gives $(n-1)\,\mathcal N_p=-2\varepsilon^2/(X_0^2-\varepsilon^2)<0$ at $t=0$, and since $v=0$ at that instant, Eq.~\eqref{eq:concavity} gives the same negative value $-\varepsilon^2/(X_0^2-\varepsilon^2)$ for every transverse direction. The axis directions are not saturated because $H_1$ and $H_2$ are nonzero on open intervals, other support planes violate the outside-support condition of Appendix~\ref{Sec:Diagonal}, and full support would force staticity.

A single saturated line allows non-diagonal evolution. Let $b(t)$ be smooth, bounded and nonconstant, for instance periodic, and set
\begin{equation}
 h=\bigl(dx+b(t)\,dz\bigr)^2+dy^2+dz^2 .
 \label{eq:onelineexample}
\end{equation}
In the notation of \eqref{eq:normalform}, $a=1$ and $v=\dot b/2$, so Lemma~\ref{lem:energy} gives $8\pi G\rho+\Lambda=-\dot b^2/4$, and condition~(ii) of Theorem~\ref{thm:main} fails wherever $\dot b\ne0$. Here $\det h=1$, $h$ is uniformly bounded, and null and timelike completeness follow as for the periodic model. Since $\dot h=\dot b\,dz\odot(dx+b\,dz)$, the line $[dz]$ is saturated. At each time with $\dot b\ne0$, Proposition~\ref{prop:nothirdline} allows exactly one further instantaneous line, $[dx+b(t)\,dz]$, which moves with $b$; so $[dz]$ is the only saturated line. If $b(t_0)=0$ at some time, as for a periodic $b$ with zero mean, then $h(t_0)=\mathsf 1$, and a constant transformation diagonalizing the whole family would make the matrices $h(t)$ commute, but
\begin{equation}
 [h(b_1),h(b_2)]=b_1b_2(b_2-b_1)\bigl(e_1e_3^{\mathsf T}-e_3e_1^{\mathsf T}\bigr)\ne0
 \label{eq:noncommuting}
\end{equation}
for distinct nonzero values $b_1,b_2$ of $b$, where $h(b)$ denotes the matrix of \eqref{eq:onelineexample} in the coordinates $(x,y,z)$. Equation~\eqref{eq:concavityN} gives $(n-1)\mathcal N_{dz}=-\tfrac32\dot b^2$.

\subsection{Compact pulses}

Compactly supported deformations realize the same two-line class with a static exterior. Let $f\in C_c^\infty(\mathbb R)$ be nonzero and choose $X_0,Y_0,c>0$ with $X:=X_0+f>0$ and $Y:=Y_0-f>0$. The metric
\begin{equation}
 h(t)=\operatorname{diag}\bigl(X(t),Y(t),c^2\bigr)
 \label{eq:compactpulse}
\end{equation}
is constant outside a compact interval, so it is null and timelike complete, the independent limit exists and
\begin{equation}
 I_p^{\mathrm{ind}}=-J_p^{\mathrm{tot}}\in(-\infty,0] .
 \label{eq:compactpulsebudget}
\end{equation}
The rays \eqref{eq:sharpmomenta} are again exactly the saturated set, and every other direction has a strictly negative, finite integral. Replacing $f(t)$ by $f(t/T)$ keeps the amplitude and sends $J_p^{\mathrm{tot}}$ to $J_p^{\mathrm{tot}}/T$, so a small deficit constrains the amplitude only when the duration is controlled. For a weak pulse, take $X_0=Y_0=c=1$ and $f=\varepsilon g(t/T)$ with fixed $g\in C_c^\infty$, normalize $|p|=1$ in the static exterior, and write $\hat n=p$ and $C_g=\int(g')^2ds$. Then, uniformly in $\hat n$,
\begin{equation}
 \begin{aligned}
 J_{\hat n}^{\mathrm{tot}}&=\frac{\varepsilon^2 C_g}{4T}
 \left[2\hat n_3^2+(\hat n_1^2-\hat n_2^2)^2\right]
 +O(\varepsilon^3/T),\\
 \langle J_{\hat n}^{\mathrm{tot}}\rangle_{S^2}
 &=\frac{7\varepsilon^2 C_g}{30T}+O(\varepsilon^3/T),
 \end{aligned}
 \label{eq:compactpulseangular}
\end{equation}
in agreement with the sky average \eqref{eq:taxaverage} for the shear $\sigma=\tfrac12\dot f\operatorname{diag}(1,-1,0)$. The angular factor is $\operatorname{tr}(P_{\hat n}S_0P_{\hat n}S_0)$ with $S_0=\operatorname{diag}(1,-1,0)$ and $P_{\hat n}=\mathsf 1-\hat n\hat n^{\mathsf T}$, and its average on the sphere is $14/15$.

\subsection{Finite episodes and the three-direction budget}

For an anisotropic episode of finite duration between static epochs, every term of the identity is finite, $I_p^{\mathrm{ind}}=-J_p^{\mathrm{tot}}$, and the deficits control the geometry. Choose constant spatial coordinates in which the static metric before the episode is $\mathsf 1_n$, and normalize the covectors by $|p_a|=1$ there, so that $\omega_{p_a}=1$ before the episode.

\begin{proposition}[Three-direction budget]
\label{prop:quantitative}
Let $h(t)$ be smooth on $\mathbb R$, equal to $\mathsf 1_n$ outside $[t_-,t_+]$, and satisfy $h(t)\preceq M\mathsf 1_n$. Set $T=t_+-t_-$. Choose three pairwise distinct unit-covector lines $[p_a]$, $a=1,2,3$, let $V_a$ have Euclidean orthonormal columns spanning $\ker p_a$, and define
\begin{equation}
 c_{\mathcal P}:=\min_{\substack{S=S^{\mathsf T}\\\|S\|_{\mathrm F}=1}}
       \sum_{a=1}^3\|V_a^{\mathsf T}SV_a\|_{\mathrm F}^2 ,
 \label{eq:directionconstant}
\end{equation}
with $\|\cdot\|_{\mathrm F}$ the Frobenius norm in the reference coordinates. Then $c_{\mathcal P}>0$, every null geodesic is complete, and $I_{p_a}^{\mathrm{ind}}=-J_{p_a}^{\mathrm{tot}}$. If $I_{p_a}^{\mathrm{ind}}\ge-\delta_a$ with $\delta_a\ge0$, then
\begin{equation}
 \begin{aligned}
 \int_{\mathbb R}\|\dot h(t)\|_{\mathrm F}^2\,dt
   &\le\frac{4M^{5/2}}{c_{\mathcal P}}\sum_{a=1}^3\delta_a,\\
 \sup_{t\in\mathbb R}\|h(t)-\mathsf 1_n\|_{\mathrm F}^2
   &\le\frac{TM^{5/2}}{c_{\mathcal P}}\sum_{a=1}^3\delta_a .
 \end{aligned}
 \label{eq:metricenergybudget}
\end{equation}
For the three coordinate covectors in $n=3$, $c_{\mathcal P}=1$. If instead $h$ is constant for $t\ge t_+$ with $h(t_+)$ possibly different from $\mathsf 1_n$, the first bound in \eqref{eq:metricenergybudget} still holds, and
\begin{equation}
 \|h(t_+)-\mathsf 1_n\|_{\mathrm F}^2\le\frac{4TM^{5/2}}{c_{\mathcal P}}\sum_{a=1}^3\delta_a .
 \label{eq:stepbudget}
\end{equation}
\end{proposition}

This is the finite-amplitude, quantitative form of Corollary~\ref{cor:averagedrigidity}. In Einstein gravity $\delta_a/(8\pi G)$ bounds the averaged null-energy deficit along $p_a$, so \eqref{eq:metricenergybudget} limits the anisotropic excursion that a given amount of averaged NEC violation can produce in a given time. The directions enter through $c_{\mathcal P}$, which tends to zero as two lines approach each other. For a finite episode with different static limits, the net metric change $h(+\infty)-h(-\infty)$ is bounded by Eq.~\eqref{eq:stepbudget}; the fossil anisotropy that such an episode could leave is discussed in Sec.~\ref{Sec:Discussion}.

\begin{proof}
Smoothness and compact support give uniform two-sided bounds on $h$, hence completeness, and $B_p$ vanishes outside $[t_-,t_+]$, so $I_p^{\mathrm{ind}}=-J_p^{\mathrm{tot}}$. If $c_{\mathcal P}=0$, a nonzero symmetric $S$ would vanish on three distinct hyperplanes $\ker p_a$. By Lemma~\ref{lem:singlehyperplane}, $S=p_1\odot\eta$. If $\eta\propto p_1$, then $S\propto p_1\otimes p_1$ vanishes on $\ker p_2$ only if $[p_2]=[p_1]$; otherwise Proposition~\ref{prop:nothirdline} allows at most the two lines $[p_1]$ and $[\eta]$. Compactness of the unit sphere of symmetric forms then gives $c_{\mathcal P}>0$. For the coordinate covectors in $n=3$, $\sum_a\|V_a^{\mathsf T}SV_a\|_{\mathrm F}^2=\|S\|_{\mathrm F}^2+\sum_iS_{ii}^2$, whose minimum on the unit sphere is one.

Write $\mathsf h_a=V_a^{\mathsf T}hV_a\preceq M\mathsf 1_{n-1}$. An $h$-orthonormal screen basis is $V_a \mathsf h_a^{-1/2}$, and \eqref{eq:Jcosmictime} becomes
\begin{equation}
 \begin{aligned}
 J_{p_a}^{\mathrm{tot}}
   &=\frac14\int\omega_{p_a}
     \|\mathsf h_a^{-1/2}\dot{\mathsf h}_a
     \mathsf h_a^{-1/2}\|_{\mathrm F}^2\,dt\\
   &\ge\frac{1}{4M^{5/2}}
     \int\|\dot{\mathsf h}_a\|_{\mathrm F}^2\,dt ,
 \end{aligned}
 \label{eq:screenbudgetbound}
\end{equation}
using $\omega_{p_a}\ge M^{-1/2}$ and the eigenvalue bound $M$ in the eigenbasis of $\mathsf h_a$. Summing over $a$ and using \eqref{eq:directionconstant} gives the first bound. For $t_-<t<t_+$, integrating $\dot h$ from either endpoint and applying Cauchy--Schwarz on the two subintervals gives $\|h(t)-\mathsf 1_n\|_{\mathrm F}^2\le(t-t_-)(t_+-t)E/T\le TE/4$ with $E=\int\|\dot h\|_{\mathrm F}^2dt$. In the step case $B_p$ still vanishes in both static regions, so $I_p^{\mathrm{ind}}=-J_p^{\mathrm{tot}}$ and the first bound are unchanged, and Cauchy--Schwarz on $[t_-,t_+]$ gives $\|h(t_+)-\mathsf 1_n\|_{\mathrm F}^2\le TE$.
\end{proof}

The upper bound $M$ can be dispensed with when the budget is small. Let $M(t)=\max\{1,\sup_{[t_-,t]}\lambda_{\max}(h)\}$. The estimate \eqref{eq:screenbudgetbound} on $[t_-,t]$, with each nonnegative partial optical integral bounded by its total budget, gives $\|h(t)-\mathsf 1_n\|_{\mathrm F}^2\le 4TM(t)^{5/2}\sum_a\delta_a/c_{\mathcal P}$, and $\lambda_{\max}(h)\le1+\|h-\mathsf 1_n\|_{\mathrm F}$. If $T\sum_a\delta_a/c_{\mathcal P}<2^{-9/2}$, the continuous nondecreasing function $M(t)$ therefore never reaches $2$, and \eqref{eq:metricenergybudget} and \eqref{eq:stepbudget} hold with $M=2$ without any a priori bound on $h$. Some such control is needed. For an isotropic pulse $h=a(t)^2\mathsf 1_n$ with $a\ge1$, peak value $A$ and fixed profile and duration, the deficits grow at most like $A^2$, whereas $\|h-\mathsf 1_n\|_{\mathrm F}^2$ grows like $A^4$.

Test fields obey a related bound. The modes of a free, massless, minimally coupled scalar are oscillators of frequency $\mathcal V\omega_p$ in the time $u=\int dt/\mathcal V$, so by \eqref{eq:Thetadef} the bound of Boonserm and Visser, Eq.~(125) of Ref.~\cite{BoonsermVisser2008}, limits the Bogoliubov coefficient across the episode to $|\beta_p|\le\sinh(\tfrac12\int|\Theta_p|\,d\lambda)$, which vanishes along saturated rays.

\bibliography{references}

\end{document}